\documentclass[11pt, a4paper]{article}
\usepackage{hyperref}
\usepackage{graphicx}
\usepackage[shortlabels]{enumitem}
\usepackage[hscale=0.75,vscale=0.8]{geometry}
\usepackage{amsmath, amsthm, amssymb, amsopn, amsbsy, mathrsfs, bbm}
\usepackage{slashed}
\usepackage[small,bf]{caption}

\usepackage{titlesec}
\usepackage{appendix}

\theoremstyle{plain}
\newtheorem{thm}{Theorem}[section]
\newtheorem{prop}[thm]{Proposition}
\newtheorem{lem}[thm]{Lemma}
\newtheorem{cor}[thm]{Corollary}

\theoremstyle{definition}
\newtheorem{defn}[thm]{Definition}

\theoremstyle{remark}
\newtheorem*{rk}{Remark}
 
\newcommand{\innerprod}[2]{\left\langle #1 \relmiddle| #2 \right\rangle}

\newcommand{\R}{\mathbb{R}}
\newcommand{\C}{\mathbb{C}}
\newcommand{\relmiddle}[1]{\mathrel{}\middle#1\mathrel{}}

\renewcommand{\d}{\ensuremath{\mathrm{d}}}
\renewcommand{\restriction}{\mathord{\upharpoonright}}

\DeclareMathOperator{\supp}{supp}
\DeclareMathOperator{\Int}{Int}

\begin{document}

\title{Non-existence of isometry-invariant Hadamard states for a Kruskal black hole in a box and for massless fields on 1+1 Minkowski spacetime with a uniformly accelerating mirror}

\author{Bernard S. Kay$^*$ and Umberto Lupo$^\dagger$ \medskip \\  {\small \emph{Department of Mathematics, University of York, York YO10 5DD, UK}} \smallskip \\ 
\small{$^*${\tt bernard.kay@york.ac.uk} $\qquad$ $^\dagger${\tt ul504@york.ac.uk}}}

\date{}

\maketitle

\begin{abstract} 
We conjecture that (when the notion of Hadamard state is suitably adapted to spacetimes with timelike boundaries) there is no isometry-invariant Hadamard state for the massive or massless covariant Klein-Gordon equation defined on the region of the Kruskal spacetime to the left of a surface of constant Schwarzschild radius in the right Schwarzschild wedge when Dirichlet boundary conditions are put on that surface.  We also prove that, with a suitable definition for `boost-invariant Hadamard state' (which we call `strongly boost-invariant globally-Hadamard') which takes into account both the existence of the timelike boundary and the special infra-red pathology of massless fields in 1+1 dimensions, there is no such state for the massless wave equation on the region of 1+1 Minkowski space to the left of an eternally uniformly accelerating mirror -- with Dirichlet boundary conditions at the mirror.  We argue that this result is significant because, as we point out, such a state does exist if there is also a symmetrically placed decelerating mirror in the left wedge (and the region to the left of this mirror is excluded from the spacetime).  We expect a similar existence result to hold for Kruskal when there are symmetrically placed spherical boxes in both right and left Schwarzschild wedges.  Our Kruskal no-go conjecture raises basic questions about the nature of the black holes in boxes considered in black hole thermodynamics.  If true, it would lend further support to the conclusion of  B.~S.~Kay `Instability of enclosed horizons', Gen.~Rel.~Grav.~{\bf 47},~1-27 (2015) (\href{http://arxiv.org/abs/1310.7395}{arXiv: 1310.7395}) that the nearest thing to a description of a black hole in equilibrium in a box in terms of a classical spacetime with quantum fields propagating on it has, for the classical spacetime, the exterior Schwarzschild solution, with the classical spacetime picture breaking down near the horizon.  Appendix~\ref{appB} to the paper points out the existence of, and partially fills, a gap in the proofs of the theorems in B.~S.~Kay and R.~M.~Wald, `Theorems on the uniqueness and thermal properties of stationary, nonsingular, quasifree states on spacetimes with a bifurcate Killing horizon', Phys.~Rep.~{\bf 207},~49-136~(1991). 
\end{abstract}

\section{Introduction}
\label{intro}

Thanks to a number of results obtained in the 1990's, it is known\footnote{Actually, while we were writing the present paper, we discovered -- see Footnote \ref{gap} --  a gap in the reasoning in \cite{kay1991theorems} which however (for all the spacetimes mentioned in this paragraph) we fill in Appendix \ref{appB} in the present paper.  So, strictly, the results we describe as previously `known' and `proven' in this paragraph and the other footnotes thereto rely on the results in Appendix \ref{appB} here as well as on the papers we cite.}  that (leaving aside some technicalities) if one quantizes a linear scalar field on a globally hyperbolic spacetime with a one-parameter group of isometries possessing a bifurcate Killing horizon, then there is at most one\footnote{\label{improvement}In fact, such a uniqueness result was proven in \cite{kay1991theorems} under the restriction that the state in question be quasi-free (with vanishing one-point function) \cite{kay1991theorems, haag1992local, bratteli1997operator2} and with the local Hadamard condition replaced by a certain global Hadamard condition (see next footnote).   However, in \cite{kay1993sufficient} a general result was obtained which enabled one to drop the quasi-free restriction while, as conjectured in \cite{kay1988quantum, gonnella1989can} and proved in \cite{radzikowski1996localtoglobal, radzikowski1996microlocal, radzikowski1992hadamard} on any globally hyperbolic spacetime, locally Hadamard states are necessarily globally Hadamard. See also Footnote \ref{gap} and Appendix \ref{appB}.} state which is invariant under those isometries and which is (locally) Hadamard.\footnote{\label{locglob}A (locally or globally) Hadamard state for a linear quantum field theory is a state whose two-point function has the (local or global) Hadamard property -- local Hadamard meaning roughly that its short distance singularity should be the appropriate generalization to a curved spacetime of the short-distance singularity of the two-point function of the vacuum state and of other physically relevant states in Minkowski space, while the global Hadamard condition on a globally hyperbolic spacetime also rules out the possibility of singularities for spacelike separated pairs of points.  For full definitions, see e.g.\ \cite{kay1991theorems} or the recent review \cite{khavkine2015algebraic}. See also the important microlocal reformulation of the global Hadamard condition in \cite{radzikowski1996microlocal} and see \cite{moretti2003comments} for spacetime dimensions other than $1+3$.}   Furthermore, for some notable cases, such as Kerr and Schwarzschild-de Sitter, it was proved in \cite{kay1991theorems} that there is \emph{no} such state.\footnote{\label{quasifreeenough}We remark that, as pointed out in \cite{kay1991theorems}, to prove such a no-go result, it suffices to prove that there is no such quasi-free state, since if there was such a state at all, the quasi-free state with the same two-point function (and zero one-point function) -- i.e.\ the `liberation' in the sense of \cite{kay1993sufficient} --  would also be such a state.}  For Kerr, this was a consequence of superradiance; for Schwarzschild-de Sitter, one argument for the no-go result was based on the fact that, should such a state exist, the Hawking temperatures associated with the black hole horizon and the cosmological horizon would be different.  Another argument relied on what, in quantum information theory, is now known as monogamy (although this notion had not yet been coined at the time).

In the present paper, we conjecture, and give heuristic arguments for, a further such non-existence result which concerns a massless or massive linear scalar field on a spacetime which one might think would represent a spherically symmetric maximally extended black hole in equilibrium in a spherical box. Namely, the region of the Kruskal spacetime to the left of a stationary hypersurface at some fixed Schwarzschild radius $R$ represented by the hyperbola in Figure \ref{kruskal} (where, as usual, each point represents a two-sphere).\footnote{Our no-go conjecture for Kruskal in a box applies equally to the part of the globally-hyperbolic region of non-extremal Reissner-Nordstr\"om spacetime to the left of a similar stationary hypersurface at fixed Schwarzschild radius $R$ but, for simplicity we shall only refer to the Kruskal case in the main text.}  I.e.\ we argue that, completing the specification of the system by imposing (say) Dirichlet boundary conditions at the box, there is no Schwarzschild-isometry invariant Hadamard state on this spacetime (when the notion of `Hadamard', usually applied to globally-hyperbolic spacetimes, is suitably adapted to the presence of a timelike boundary).  As we discuss below, this conjecture raises obvious questions about the nature of the black holes in boxes considered in the subject of black hole thermodynamics \cite{hawking1976black, gibbons1993euclidean}.

\begin{figure}
   \centering
    \includegraphics[scale=0.5]{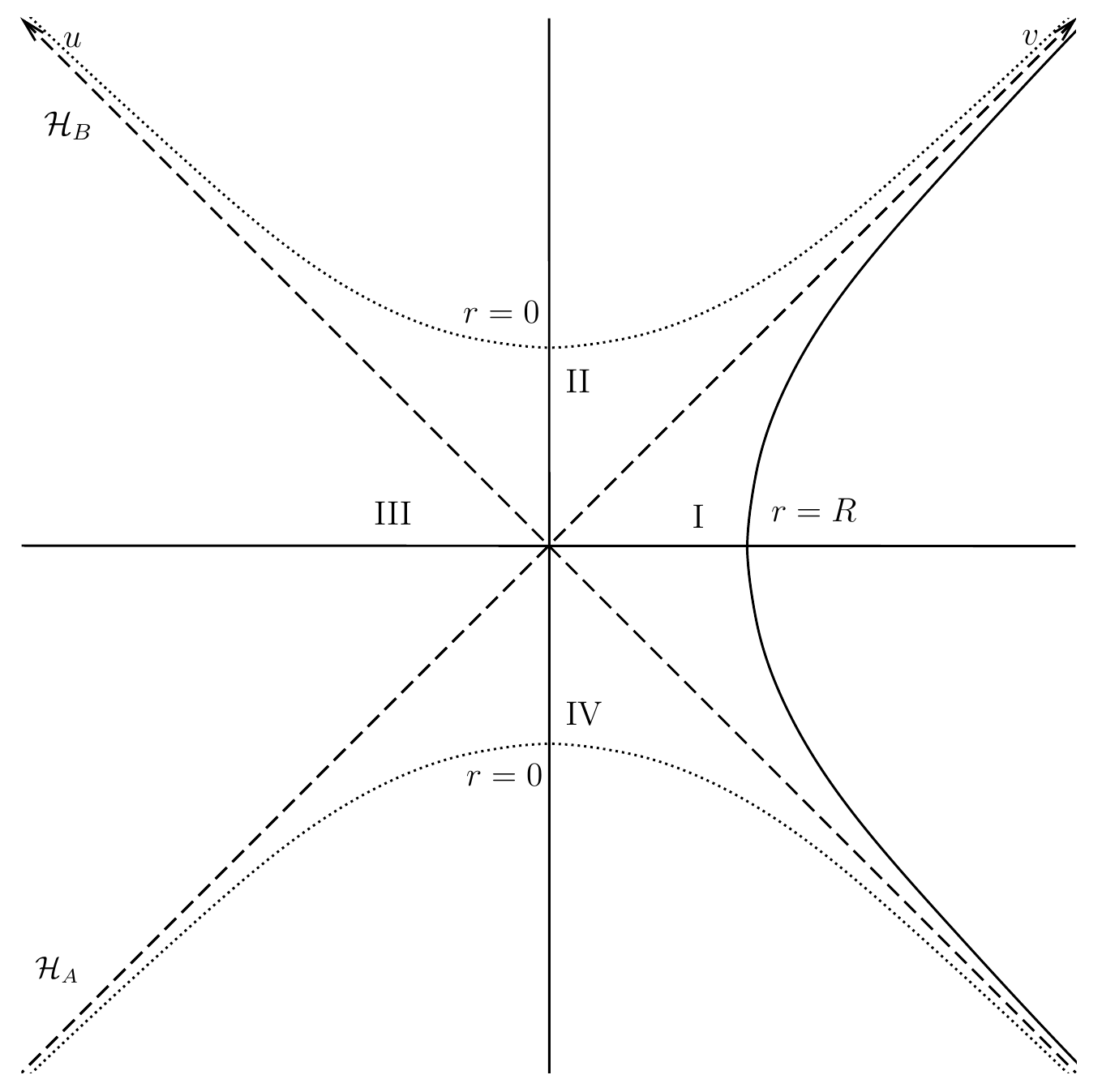}
   \caption{This is a dual purpose figure. In one interpretation, it represents the Kruskal spacetime bounded by a single box in the right wedge (region I) at $r=R$ (with $r$ the Schwarzschild coordinate and each point representing a two-sphere). In another interpretation, it represents (1+$n$)-dimensional Minkowski space to the left of a hypersurface (referred to in the text as a `mirror') at some constant Rindler spatial coordinate $r$ in the right Rindler wedge (in this case each point represents an ($n$-1)-plane). The hyperbolae in regions II and IV are only relevant to the Kruskal interpretation, in which case they portray the future and past singularities at $r=0$. \label{kruskal}}
 \end{figure}
 
 \begin{figure}
   \centering
    \includegraphics[scale=1]{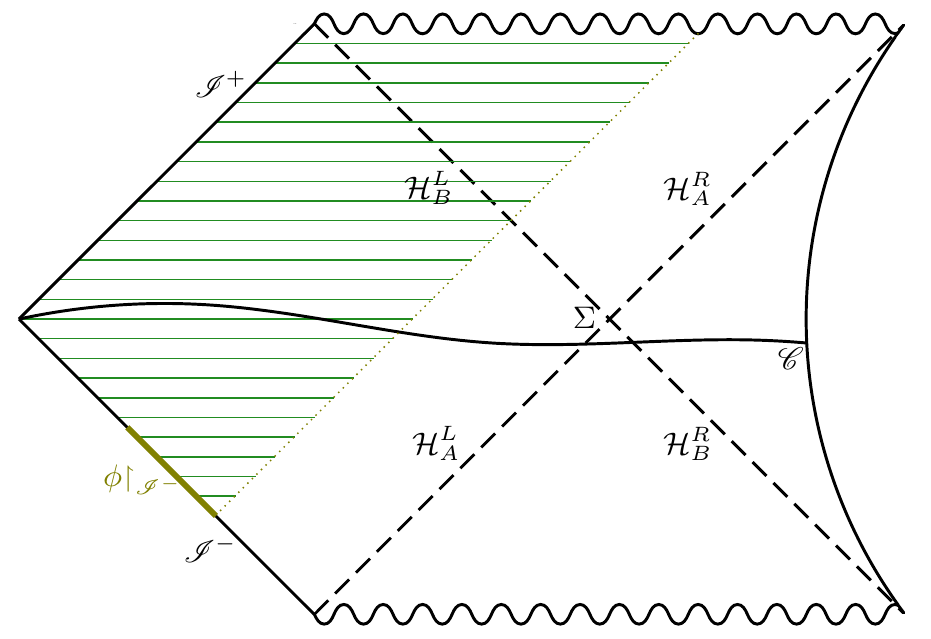}
   \caption{Penrose diagram for the region of the Kruskal spacetime bounded by a single box, cf.\ Figure \ref{kruskal}. $\mathscr{C}$ is an initial-value surface on which the Cauchy-Dirichlet problem for the Klein-Gordon equation is well-posed. The shaded area represents the support of the solution $\phi$ discussed in Sections \ref{intro} and \ref{basicidea}. \label{penrose1box}}
 \end{figure}

The basic plausible expectations about the space of classical solutions, from which we will argue for this no-go conjecture in the next section, are that, on the one hand,

\medskip

\noindent
(a) the reflection at the box in the right wedge will cause solutions which `fall entirely through' (see Section \ref{basicidea}) the right $A$-horizon (${\cal H}_A^R$ in the Penrose diagram, Figure \ref{penrose1box}) to coincide with solutions which `fall entirely through' the right $B$-horizon (${\cal H}_B^R$ in Figure \ref{penrose1box}).   

\medskip

\noindent
On the other hand, 

\medskip

\noindent
(b) there exist solutions (one such suffices for our argument) which are non-vanishing on the left $B$-horizon but which vanish on the entire $A$-horizon. 

\medskip

The plausibility of Property (b) is particularly easy to see for the massless case since, in fact, \text{any} solution, $\phi$, with non-zero Cauchy data on $\mathscr{I}^-$ (see the Penrose diagram, Figure \ref{penrose1box}) and zero Cauchy data on ${\cal H}_A$ would be expected to have a non-zero value on ${\cal H}_B$ expressing the fact that not all of the solution would be scattered back out to infinity, but rather, some of it will fall through ${\cal H}_B^L$ into the black hole.  (Whether or not this property holds obviously doesn't depend on whether or not the spacetime is cut off at a box-wall in the right wedge.)   For massless and massive fields, one can rely, instead, e.g.\ on the existence  of wave operators, $\Omega_0^\pm$ and $\Omega_1^\pm$\footnote{$\Omega_0^\pm$ maps solutions of the Klein-Gordon equation on Minkowski space into solutions on exterior Schwarzschild (identified here with our Kruskal left wedge) which resemble them at late/early times and $\Omega_1^\pm$ maps solutions of the massless `wave equation' in 1+1 Minkowski space times the bifurcation 2-sphere into solutions on exterior Schwarzschild and (as explained in \cite{dimock1987classical}) effectively solves the characteristic initial-value problem for data on the future/past horizon.} for the scattering theory on exterior Schwarzschild demonstrated in \cite{dimock1987classical, dimock1986classical} together with the expectation that the S-matrix component $(\Omega_1^+)^*\Omega_0^-$ will not be zero. In fact this is now rigorously established in the massless case in Theorem 10 of \cite{dafermos2014scattering}.\footnote{We thank Mihalis Dafermos for drawing this to our attention.}

We remark that if there is also an image box in the left wedge (located at the wedge-reflected set of spacetime points to those occupied by the right-wedge box -- below we shall refer to this as the case of two boxes) we expect that there \emph{will} exist an isometry-invariant Hadamard state on the region between the two boxes.  Indeed, we expect the latter to be a counterpart to the Hartle-Hawking-Israel state \cite{hartle1976path, israel1976thermo, sanders2015construction} in maximally extended Kruskal. Thus our no-go conjecture is reliant upon there being just one box rather than two.

Geometrically, this setup appears analogous to Minkowski spacetime (of any dimension) to the left of a hypersurface at some constant Rindler spatial coordinate in the right wedge (see Figure \ref{kruskal}), i.e.\ to the left of a uniformly accelerating mirror (assumed to be `planar' and infinitely extended in the spatial dimensions suppressed in Figure \ref{kruskal}). Here, Schwarzschild-isometry invariance is replaced by boost invariance. One might therefore think that a similar non-existence result would hold for boost-invariant Hadamard states for Klein-Gordon fields on such spacetimes.  And, in the absence of a rigorous proof of our conjecture for Kruskal, it would obviously be of interest if one could more easily give a rigorous proof of the non-existence of boost-invariant Hadamard states for some such Minkowskian system. However, Property (b) above only holds for scalar fields in Minkowski space when those fields are massless and the Minkowski space is 1+1 dimensional.  This is because, except in this special case, a solution to the Klein-Gordon equation in Minkowski space  (say with compact support on spacelike Cauchy surfaces) which vanishes on a single null plane, vanishes everywhere. See e.g.\ pages 109--110 in Section 5.1 in \cite{wald1994quantum} where this is proven for the case of massless fields and spacetime dimension greater than 2.  It is also stated there that the same statement presumably also holds for massive fields and one of us \cite{lupo2016toappear} has recently proven this.

In view of the above, and aside from making our above conjecture for the Kruskal case, the main purpose of the present paper is to prove a rigorous version of such a non-existence result for this latter 1+1 massless system with Dirichlet boundary conditions.  Even for this much simpler problem, it will turn out that we have to deal with a number of complications which arise from the well-known special infra-red pathology \cite{schroer1963infrateilchen, wightman1967introduction, streater1970fermion, kay1985double, fulling1987temperature, derezinski2006quantum} of the 1+1 massless Klein-Gordon field as well as with complications due to the presence of a boundary.  In fact, even in the absence of boundaries, because of that special infra-red pathology, there are several inequivalent mathematical notions which could be regarded as making the phrase `boost-invariant Hadamard state' precise for the massless scalar field in 1+1 Minkowski space.  What we succeed in doing (with Theorem \ref{mainthm} in Section \ref{theorem}) is to prove that, with a particular such notion, when suitably adapted to the presence of a single mirror -- namely what we call the `strongly boost-invariant globally-Hadamard' property of Definition \ref{stronglyhadamard} in Section \ref{theorem} -- then (in the presence of a single mirror) there is no state which has this property.

We believe this no-go theorem deserves to be regarded as a suitable counterpart to the no-go result we conjecture for Kruskal because, as we will also point out in Section \ref{theorem}, there \emph{does} exist a strongly boost-invariant globally-Hadamard state both in full 1+1 Minkowski space and in the case where there is a second mirror located at the wedge-reflected set of spacetime points to those occupied by the right-wedge mirror, and the region to the left of this mirror is excluded from the spacetime.  The state in the former case is a suitably defined version of the usual Minkowski vacuum state, while the state in the latter case -- which we shall call the case of two mirrors -- was constructed in \cite{kay2015instability}. Also, we think that the method of proof of our no-go theorem should provide useful lessons towards a proof of our conjecture about the Kruskal case. Note that our notion of `strongly boost-invariant globally-Hadamard' makes precise the notion of `boost-invariant \emph{global} Hadamard state' since, for reasons we will explain in Section \ref{infra-red}, we do not know if a local-to-global result (see Footnote \ref{locglob}) applies in the 1+1 massless case.

Our conjecture in the Kruskal case has an obvious application to understanding the nature of the idealized black holes in boxes which play a basic role in black hole thermodynamics \cite{hawking1976black, gibbons1993euclidean}. A natural question is whether a black hole in equilibrium in a box\footnote{Here we leave aside the issue that a Schwarzschild black hole in equilibrium in a box is believed to be thermodynamically unstable \cite{hawking1976black}.  We remark that, as explained in \cite{kay2015instability},  the Schwarzschild anti-de Sitter spacetime (where, for certain values of the parameters, one has thermodynamic stability) is, when maximally extended, analogous to the region of Kruskal between two boxes -- i.e. what we call in the main text, `case (B)' -- and thus the results of the present paper are not relevant to it; however the results in  \cite{kay2015instability} suggest that this maximal extension also suffers from the same problems as case (B) for Schwarzschild black holes and therefore that a physical Schwarzschild anti-de Sitter black hole will be a single Schwarzschild-anti de Sitter wedge with a non-classically describable region near the horizon analogously to what we argue for Schwarzschild black holes.} has a semiclassical description in terms of a fixed Lorentzian classical spacetime together with a Hadamard state of a quantum field defined on it -- where both the classical spacetime and the Hadamard state are isometry-invariant.  Amongst the various possibilities one can imagine for the background spacetime, and ignoring back reaction, one might consider the following three: (A) the region of Kruskal to the left of a single box as in Figures \ref{kruskal} and \ref{penrose1box}; (B) the region of Kruskal between two boxes as in Figure \ref{penrose2boxes}; (C) the region of exterior Schwarzschild alone to the left of a single box (i.e.\ the right wedge of any of the figures \ref{kruskal}, \ref{penrose1box} or \ref{penrose2boxes}). 
An earlier paper \cite{kay2015instability} of one of us argued that both (A) and (B) should be ruled out due to the existence of classical and/or quantum small perturbations such that, as a consequence of reflection at the box, their (renormalized) stress-energy grows arbitrarily large near the future horizon(s) and/or near the bifurcation surface and argued in favour of (C) with the proviso that the region near the horizon be considered to be essentially quantum-gravitational and non-classically describable rather as envisaged in `t Hooft's `brick wall' model \cite{thooft1985quantum}.  However the arguments against (A) in \cite{kay2015instability} were less strong than the arguments against (B).  Our conjectured no-go theorem, if true, tells us that, on the background (A), no isometry-invariant Hadamard state is possible and this reinforces our reasons for rejecting (A).

\begin{figure}
   \centering
    \includegraphics[scale=1]{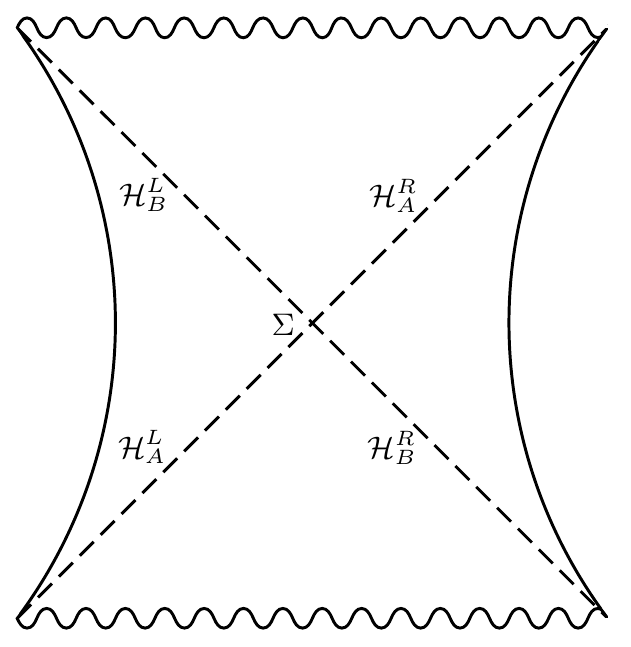}
   \caption{Penrose diagram for the region of the Kruskal spacetime bounded by two boxes, possibility (B) in Section \ref{intro}. We conjecture that a `Hartle-Hawking-Israel--like' state exists for the Klein-Gordon field on this spacetime when Dirichlet conditions are imposed at the boundary. \label{penrose2boxes}}
 \end{figure}

It is also of interest to compare our no-go result for the massless scalar field in 1+1 Minkowski with claims made in the literature (see e.g.\ \cite{fulling1976radiation, davies1977radiation, birrell1984quantum}) concerning radiation by accelerating mirrors in 1+1 dimensions.   As pointed out in that work, a mirror which starts out inertial -- with the state of the field the initial vacuum state -- and later undergoes uniform acceleration doesn't radiate during the period of uniform acceleration. This might seem to suggest that there would be a quantum state of the field such that an eternally accelerating mirror wouldn't radiate at all and that might, in its turn, seem to suggest that there would exist a boost-invariant Hadamard quantum state.  And one might think that there would in fact exist a strongly boost-invariant globally-Hadamard state in the sense of the present paper.  But we prove that there isn't one; for there to be such a boost-invariant Hadamard state, it would seem to be required for there to be a symmetrically placed uniformly decelerating image mirror in the left wedge.

An outline of the structure of the rest of the paper is given in the paragraph preceding Equation (\ref{sympK}) in Section \ref{basicidea}.

There are two Appendices.   The purpose of Appendix \ref{appA} is explained in the above-mentioned paragraph.  Appendix \ref{appB} points out the existence of, and partially fills, a gap in the arguments in the 1991 paper \cite{kay1991theorems} of B.S.~Kay and R.M.~Wald.  It is included here because (see Footnote 19) the gap became apparent while we were doing this work.  However  its content is logically independent of the rest of the paper.

\section{Basic idea of our argument for our Kruskal conjecture and for our $1$$+$$1$ no-go theorem}\label{basicidea}

We next wish to explain the basic idea behind both our no-go conjecture for (massive or massless) Klein-Gordon on Kruskal and our proof of our analogous no-go result for the massless 1+1 Minkowski one-mirror system.  In Kruskal we take our equation to be 
\begin{equation}
\label{kleingordon} 
P \phi = (\square_g + m^2)\phi = 0
\end{equation}
where $m$ is a non-negative mass.  (One could add a term proportional to the Ricci scalar, $R$, to $m^2$, but this of course vanishes in Kruskal.)  In our 1+1 Minkowskian theorem we insist that $m$ be zero.

In both cases, we rely on the well-posedness of the Cauchy problem for (\ref{kleingordon}) when supplemented by Dirichlet boundary conditions at the box/mirror.  Of course, neither the region of Kruskal to the left of our box, nor the region of 1+1 Minkowski space to the left of our mirror are globally hyperbolic and thus neither have Cauchy surfaces in the strict sense.  However, with our boundary conditions on the box/mirror, one expects the Cauchy problem to be well posed, at least in the sense of uniqueness, for data on initial-value surfaces which are the restrictions, to the region to the left of the box/mirror, of Cauchy surfaces for the whole of Kruskal/Minkowski. Indeed, this can easily be verified in the 1+1 Minkowski case; for the Kruskal case we expect a suitable extension of known results on the mixed Cauchy-Dirichlet problem (see e.g.\ Theorem 24.1.1 in \cite{hormander2007analysis} or the monograph \cite{gindikin1996mixed}) to apply. And it will still to be possible to define, in each case, the space $S$ of smooth (real-valued) solutions of this mixed Cauchy-Dirichlet problem whose restriction to all such initial-value surfaces\footnote{These initial-value surfaces should be understood to contain the relevant boundary points and therefore not as being entirely contained in the interior of the spacetime.} has compact support, analogously to the definition of $S$ in \cite{kay1991theorems}. And this space will be equipped with a manifestly antisymmetric bilinear form $\sigma$ defined, in terms of an arbitrary (possibly partially null) smooth initial-value surface $\mathscr{C}$, by
\begin{equation}
\label{symplformdefinition} 
\sigma(\phi_1, \phi_2) := \int_\mathscr{C} n_a j^a[\phi_1, \phi_2] \, \d \mu_\mathscr{C}, 
\end{equation}
where $j^a[\phi_1, \phi_2] := \phi_1 \nabla^a \phi_2 - \phi_2 \nabla^a \phi_1$, $\mathscr{C}$ is given the induced orientation as the boundary of $J^-(\mathscr{C}),$\footnote{I.e.\ the boundary orientation for which Stokes' Theorem applies.  Here, $J^\pm$ of a subset of a spacetime denotes its causal future/past \cite{hawking1973large}} and the forms $n$ and $\d \mu_{\mathscr{C}}$ are such that, on $\mathscr{C}$, $n \wedge \d \mu_\mathscr{C}$ equals the volume form $\d \mu_g$ induced by the spacetime metric. The independence of the right-hand side of Equation (\ref{symplformdefinition}) from the initial-value surface $\mathscr{C}$ is a consequence, using Gauss' theorem, of the fact that $\nabla_a j^a[\phi_1, \phi_2] = 0$ whenever $\phi_1$ and $\phi_2$ are solutions to Equation (\ref{kleingordon}), together with the fact that, due to the Dirichlet boundary conditions, no boundary terms arise from integrating along the spacetime boundary. One expects that, once a full characterization for the allowed initial data for solutions in $S$ is available, it will be possible to show that $\sigma$ is in fact non-degenerate on $S$, and therefore a symplectic form.

Similarly to in \cite{kay1991theorems} -- and proceeding, in the case of Kruskal, under the same fiction explained in the note added in proof at the end of \cite{kay1991theorems} (see the discussion at the end of this section) -- an important role will be played by `subspaces', $S_A$ and $S_B$, of $S$ which consist of solutions which `fall entirely through' the $A$- and $B$-horizons ${\cal H}_A$ and ${\cal H}_B$ respectively. Precisely, a solution $\phi$ belongs to $S_A$ if its support intersects ${\cal H}_A$ in a compact set and if $\phi$ vanishes outside the union of the causal past and causal future of this set (and one defines $S_B$ analogously).  For a massless scalar field in 1+1 Minkowski space without any mirrors, $S_B$ would consist of right-moving solutions and $S_A$ of left-moving solutions.  When we have our mirror in the right wedge, $S_B$ consists of solutions which are right-moving to the causal past of the $B$-horizon, and $S_A$ consists of solutions which are left-moving to the causal future of the $A$-horizon as explained in more detail in Section \ref{1+1classicaltheory}.  We also define $S_A^R$ to consist of solutions in $S_A$ whose restrictions to the $A$-horizon are compactly supported to the right of (and strictly away from) the bifurcation surface, and also define $S_A^L$, $S_B^R$ and $S_B^L$ similarly with obvious changes.

In Appendix \ref{appA}, we will recall the general theory of the quantization of linear Bose systems via the so-called Weyl-algebra approach. In particular, we will review the standard definitions for the notions of \emph{state}, \emph{quasifree state} and \emph{one-particle structure}. In Section \ref{quantumkleingordon}, we will recall how this theory is applied to the case of Klein-Gordon fields on general globally hyperbolic spacetimes, where the class of Hadamard states (see Footnote \ref{locglob}) plays a special role, and we will sketch a strategy for adapting this theory to situations with timelike boundaries so as to properly define the notion of `Hadamard state' and, thereby, to be able to formulate in a precise way our conjecture that there is no isometry invariant Hadamard state on Kruskal in a box.   Then Section \ref{1+1} will show how to implement this strategy for massless fields on 1+1 Minkowski with a mirror in a way which also copes with the special infra-red pathology, thereby enabling us both to properly formulate and prove our no-go theorem.  For us to explain the basic idea behind our conjecture and theorem in the present section, however, all that we shall rely on are the following two facts:
\begin{itemize}[topsep=1ex,itemsep=1ex,partopsep=0ex,parsep=0ex]
\item First, just as in the globally hyperbolic case mentioned in Footnote \ref{quasifreeenough}, to show that there is no isometry-invariant Hadamard state, it suffices to show there is no isometry-invariant \emph{quasi-free} Hadamard state (with zero one-point function), see Appendix \ref{appA}.   
\item Second, as explained in Appendix \ref{appA}, to every quasi-free state of the theory there corresponds a \emph{one-particle structure}, $(K, \mathscr{H})$.   That is a, Hilbert space (the \emph{one-particle Hilbert space}), $\mathscr{H}$, and a real-linear map, $K:S \rightarrow \mathscr{H}$, such that $KS + i KS$ is dense in $\mathscr{H}$,
which is \emph{symplectic} in the sense that
\begin{equation}
\label{sympK}
2 \mathrm{Im}\innerprod{K\phi_1}{K\phi_2} = \sigma(\phi_1,\phi_2)
\end{equation}
for all pairs of classical solutions, $\phi_1, \phi_2\in S$.
\end{itemize}

Furthermore, and similarly to Kruskal without a box or (1+3)-dimensional Minkowski without a mirror, we expect that the existence of an isometry-invariant Hadamard state for Kruskal with our box implies, by similar arguments to those given in \cite{kay1991theorems}, the following explicit formula for $\innerprod{K\phi_B^1}{K\phi_B^2}$ for any $\phi_B^1,\phi_B^2 \in S_B$:
\begin{equation}
\label{horizontwopoint}
\innerprod{K\phi_B^1}{K\phi_B^2} = - {1\over \pi} \lim_{\varepsilon\rightarrow 0^+} \int 
{f_1(u_1, s)f_2(u_2, s)\over (u_1-u_2-i\varepsilon)^2}\, \d u_1\, \d u_2 \,d^2 s,
\end{equation}
where $f_1$ is the restriction of $\phi_B^1$ and $f_2$ the restriction of $\phi_B^2$ to the $B$-horizon, and this is coordinatized in the usual way by affine parameter, $u$,\footnote{Aside from having the opposite signature convention to \cite{kay1991theorems}, we (and also \cite{kay2015instability}) differ from \cite{kay1991theorems} by denoting affine parameter on our horizons by $u$ and $v$, rather than $U$ and $V$.} and the usual set of angular variables, denoted by $s$, and the integration can be thought of as over two copies of the real line and one copy of the bifurcation sphere. 

For our massless scalar field in 1+1 Minkowski with a mirror, it turns out that the existence of an 
isometry-invariant state which is Hadamard in the precise sense we will define (i.e.\ the `strongly boost-invariant globally-Hadamard' property of Definition \ref{stronglyhadamard} in Section \ref{theorem}) entails a similar formula, with the dependence on $s$ and the integration over $s$ removed.  And of course there will be a similar formula, for $\phi_A^1$ and $\phi_A^2$ and the $A$-horizon.

As discussed in \cite{kay1991theorems} (cf.\ Equation (1.1) there; we refer also to Observation 6.1 and Proposition 7.2 in \cite{dimock1987classical}), Equation (\ref{horizontwopoint}) tells us that the restriction of the two-point function for the $u$ derivative of the field to the $B$-horizon can be identified (up to a trivial dependence on $s$) with the restriction of the two-point function for the $u$ derivative of a free massless real scalar field in 1+1 Minkowski space (without a mirror) to the null line $t=-x$, where $u$ is now identified with $t-x$, and where $t$ and $x$ are the usual Minkowski coordinates.  In view of this (or directly from the formula) one can conclude (see again \cite{kay1991theorems}) the following crucial facts\footnote{Actually in our proof of our no-go theorem, i.e.\ of Theorem \ref{mainthm} in Section \ref{theorem}, facts (A) and (B) about the one-particle structure $(K, \mathscr{H})$ are arrived at by directly relating it to the one-particle structure $(K_\mathbb{M}, \mathscr{H}_\mathbb{M})$ associated to the vacuum state, $\omega_\mathbb{M}$, on the `physical' Weyl algebra for the  massless wave equation in (1+1)-Minkowski space by a somewhat different version of the argument which doesn't (need to) refer to the formula (\ref{horizontwopoint}).}
\begin{enumerate}[label=(\Alph*)]
\item $KS_A$ and $KS_B$ are dense in complex-linear subspaces $\mathscr{H}_A$ and $\mathscr{H}_B$ of $\mathscr{H}$ (respectively).   As explained in Appendix A of \cite{kay1991theorems}, and reproduced in Appendix \ref{appA} to the present paper as Proposition \ref{purity}, this is equivalent to the fact that the state restricted to fields `symplectically smeared' with solutions in either $S_A$ or $S_B$ is a pure state.  In the special case of 1+1 Minkowski (without a mirror) it corresponds to the fact that the Minkowski vacuum is a pure state when restricted to either the left or right-moving sector.
\item $KS_A^R + iKS_A^R$ is dense in $\mathscr{H}_A$ and $KS_B^R + iKS_B^R$ is dense in $\mathscr{H}_B$. This corresponds to the fact that the (massless) 1+1 Minkowski vacuum, restricted to sums of products of (derivatives of) fields restricted to a single null line has the Reeh-Schlieder property \cite{streater2000pct} for fields localised on a half null-line. Cf.\ Proposition \ref{reehschlieder} in Appendix \ref{appA}.
\end{enumerate}

We are now in a position to explain the basic idea behind both our hoped-for proof of our no-go conjecture for Kruskal in a box and our proof of our no-go theorem for our massless field in 1+1 Minkowski with a mirror.

First we point out that, for the 1+1 Minkowski case, the two `basic plausible expectations about the space of classical solutions' discussed in Section \ref{intro} are both satisfied, and may be reformulated in terms of our subspaces of solutions as follows:
\begin{enumerate}[label=(\alph*)]
\item $S_A^R=S_B^R$;
\item There exists a $\phi \in S$ such that $\sigma(\phi, \phi_B^L) \ne 0$ for some $\phi_B^L \in S_B^L$, but for which $\sigma(\phi, \phi_A) =0$ for all $\phi_A\in S_A$.
\end{enumerate}

Combining the (purely classical) statements in (a) and (b) with (A) and (B) above quickly leads to a contradiction, as we will now explain.  By the first part of (b) and Equation (\ref{sympK}), $K\phi$ cannot be orthogonal to $KS_B^L$ and hence, \emph{a fortiori} it cannot be orthogonal to $KS_B$ -- so, by (A), it cannot be orthogonal to $\mathscr{H}_B$.  On the other hand, Equation (\ref{sympK}) and the last part of (b), together with (A), imply that $K\phi$ is orthogonal (in the Hilbert space $\mathscr{H}$) to $\mathscr{H}_A$.  To see this, we will use the following general observation: If $\mathscr{H}$ is a complex Hilbert space, and $\mathscr{K} \subseteq \mathscr{H}$ is a real-linear subspace whose closure is \emph{complex}-linear, then, for any $\Phi \in \mathscr{H}$, $\innerprod{\Phi}{\mathscr{K}} = \{0\}$ if and only if $\mathrm{Im}\innerprod{\Phi}{\mathscr{K}} = \{0\}$ if and only if $\mathrm{Re}\innerprod{\Phi}{\mathscr{K}} = \{0\}$. [Proof of first `if': Suppose that $\mathrm{Im}\innerprod{\Phi}{\mathscr{K}} = \{0\}$.  Note that $\mathrm{Re}\innerprod{\Phi}{\mathscr{K}} = \mathrm{Im}\innerprod{\Phi}{i\mathscr{K}}$.  Under the assumptions on $\mathscr{K}$, $i \mathscr{K} \subseteq \overline{\mathscr{K}}$, whereupon a simple limit argument shows that $\mathrm{Im}\innerprod{\Phi}{\mathscr{K}} = \{0\}  \Longrightarrow \mathrm{Im}\innerprod{\Phi}{i \mathscr{K}} = \{0\}$ and we are done.  The proof of the second `if' is analogous.]  By (B), to say that $K \phi \perp KS_A$ is tantamount to saying that it is orthogonal to $KS_A^R +iKS_A^R$.  But, by (a), this is the same thing as saying that it is orthogonal to $KS_B^R +iKS_B^R$, which, by (B), has the same closure as $KS_B$, namely $\mathscr{H}_B$.  Thus, on the assumption that there exists a stationary Hadamard state, $K\phi$ is both not orthogonal to $\mathscr{H}_B$ and orthogonal to $\mathscr{H}_B$ -- a contradiction.

For Kruskal in a box, Property (a) above cannot strictly hold since we would expect a solution which falls entirely through the right $B$-horizon to have a restriction to the right $A$-horizon which fails to be supported away from the bifurcation point and moreover we would expect it to fail to be compactly supported, but rather to have a tail at large $v$. However, we conjecture that the closure in $\mathscr{H}$ of $KS_A^R$ will equal the closure in $\mathscr{H}$ of $KS_B^R$ (or rather an appropriate substitute for this statement will hold when one removes the fiction we referred to above and discuss further below).  It is easy to see that this `closure conjecture' would immediately lead to the same contradiction.  

The fiction we referred to above concerns an error in the original version of \cite{kay1991theorems} which we have also (knowingly) made above. As was pointed out in the note added in proof in that paper, the notion of `$C^\infty$ solutions which fall entirely through one of the horizons', as in the apparent `definitions' of $S_A$ etc.\ in that paper and above in the Kruskal case, is problematic since a solution which actually falls entirely through one of the horizons in the sense explained above cannot be $C^\infty$ -- smoothness failing when one crosses from one side of the horizon to the other. The note added in proof of \cite{kay1991theorems} showed how one can repair this error while maintaining the spirit of the basic arguments there by working with a certain class of solutions (which are everywhere $C^2$) and end up with rigorous results with essentially the same physical content as those originally announced.   In particular, the no-go results in that paper continue to hold with thus-corrected arguments. We remark that, in a recent paper \cite{sanders2015construction}, K.\ Sanders has pointed out that some of the arguments in the note added in proof may possibly be made simpler using an approach \cite{hormander1990remark} to the characteristic initial-value problem due to H{\"o}rmander (see also \cite{baer2015initial}). However, to our knowledge, this idea has not been pursued in detail.  Some new alternative ways to deal with some of the technical issues in the note added in proof in \cite{kay1991theorems} are also indicated in Appendix \ref{appB} here.

Clearly, in the case of Kruskal, what we have written above, while we find it highly plausible, falls considerably short of being a rigorously stated theorem and proof.  To have a rigorously stated theorem one would need to show that the expectations mentioned in Section \ref{timelikeboundaries} below hold so that the strategy we sketch there for defining what is meant by a Hadamard state can be implemented.  And then to turn the above-explained idea for a proof into a rigorous proof one would need to remove the above fiction, presumably with similar methods to those introduced in the note added in proof in \cite{kay1991theorems}, prove the above `closure conjecture' or some effective replacement for it, and justify in detail the various statements made above which were described as `expectations'.  As we anticipated in the Introduction, in the absence of all that, what we can and do provide, in Section \ref{1+1}, is a rigorous formulation and proof of our no-go result for a massless field in 1+1 Minkowski with a mirror.

\section{Quantization of Klein-Gordon quantum fields}\label{quantumkleingordon}

\subsection{Globally hyperbolic case}\label{kleingordonglobhyp}

Let $(M, g)$ be an oriented, time-oriented, globally hyperbolic spacetime of dimension $1+n$. (We adopt the signature convention $(+, -, \ldots, -)$ for the metric.) We recall that the vector space, which we will denote by $S$, of `regular' real-valued classical solutions to the Klein-Gordon equation, Equation (\ref{kleingordon}), is naturally equipped with a linear symplectic structure. Explicitly, the symplectic product of any two such solutions $\phi_1, \phi_2$ is defined by Equation (\ref{symplformdefinition}), where $\mathscr{C}$ is any smooth Cauchy surface, and by `regular' we mean that $\phi \in S$ should be (a) smooth and (b) `spacelike compact', i.e.\ compactly supported when restricted to any Cauchy surface (equivalently, $\mathrm{supp} \, \phi \subset J(K)$\footnote{Throughout this paper, given a subset $S$ of a spacetime, $J(S)$ denotes $J^+(S)\cup J^-(S)$ where $J^\pm(S)$ is the causal future/past \cite{hawking1973large} of $S$.} for some compact set $K$). Denoting by $P$  the Klein-Gordon differential operator as in Equation (\ref{kleingordon}), and by $C_{\mathrm{sc}}^\infty(M)$ the space of real-valued, smooth, spacelike compact functions on $M$, this amounts to defining $S$ as $\ker (P \restriction_{C_{\mathrm{sc}}^\infty(M)})$. 

Standard theory \cite{baer2007wave} guarantees that the Cauchy problem for Equation (\ref{kleingordon}) in such a spacetime is well-posed, and that there exist retarded/advanced fundamental solutions (Green's functions) $E^\pm : C_0^\infty(M) \to C^\infty(M)$ which are uniquely determined by requiring that they
\begin{enumerate}[(i)]
\item be right inverses to $P$ and left inverses to $P\restriction_{C_0^\infty(M)}$,
\item satisfy the support properties $\supp (E^\pm F) \subseteq J^\pm (\supp F)$ $\forall \ F \in C_0^\infty(M)$.
\end{enumerate}
Letting $E := E^- - E^+ : C_0^\infty(M) \to C^\infty(M)$, it is evident that $E$ maps test functions to elements of the space $S$ defined above. We call $E$ the \emph{causal propagator} of the theory since $\supp (E F) \subseteq J(\supp F)$. Furthermore, the sequence of vector spaces
\begin{equation}\label{exactsequence}
\{ 0 \} \; \longrightarrow C_0^\infty(M) \; \stackrel{P}{\longrightarrow} \; C_0^\infty(M) \; \stackrel{E}{\longrightarrow} \; C_{\mathrm{sc}}^\infty(M) \; \stackrel{P}{\longrightarrow} \; C_{\mathrm{sc}}^\infty(M)
\end{equation}
is exact, implying in particular that $E$ is onto $S$, that $\ker E = P[C_0^\infty(M)]$ and therefore also that $S \cong C_0^{\infty}(M)/P[C_0^\infty(M)]$.  One also verifies that, for any $\phi_1, \phi_2 \in S$,
\begin{equation}\label{symplformcovariant}
\sigma(\phi_1, \phi_2) = \int_M F_1 \phi_2 \, \d \mu_g = \int_{M} F_1 (E F_2) \, \d \mu_g =: E(F_1, F_2),
\end{equation}
where $\d \mu_g$ denotes the metric volume form, and $F_1, F_2 \in C_0^\infty(M)$ are such that $E F_1 = \phi_1$ and $E F_2 = \phi_2$. 

The \emph{Weyl algebra} recipe for quantization of general linear systems outlined in Appendix \ref{appA} can now be straightforwardly applied to $(S, \sigma)$, thus yielding a Weyl algebra of canonical commutation relations ${\cal A} = \mathscr{W}(S, \sigma)$. In view of the existence of the causal propagator $E$ relating test functions to solutions, if $\omega$ is a $C^2$ state on ${\cal A}$, then its two-point function $\lambda_2$ (see Appendix \ref{appA}) induces a bidistribution\footnote{Henceforth, for a manifold (without boundary) $N$, we use the word `bidistribution on $N$' to simply indicate a bilinear functional $C_0^\infty(N) \times C_0^\infty(N) \to \C$, without any continuity requirements.} on $M$ defined for all test functions $F_1, F_2$ by
\begin{equation}\label{spacetimetwoptfn}
\Lambda (F_1, F_2) = \lambda_2(EF_1, EF_2).
\end{equation}
We will henceforth refer to $\lambda_2$ as the `symplectically smeared two-point function' and to $\Lambda$ as the `spacetime smeared two-point function'. In view of the general properties of $C^2$ states listed in Appendix \ref{appA}, of the sequence (\ref{exactsequence}) and of Equation (\ref{symplformcovariant}), $\Lambda$ will satisfy for all $F_1, F_2 \in C_0^\infty(M)$:
\begin{enumerate}
 \item (\emph{Commutation relations}) $\mathrm{Im} [\Lambda(F_1, F_2)] = [\Lambda(F_1, F_2) - \Lambda(F_2, F_1)]/(2i) = E(F_1,F_2)/2$;
 \item (\emph{Positivity}) $\mathrm{Re} \Lambda$ has analogous symmetry and positivity properties to (i)--(ii) in Appendix \ref{appA} (with $\sigma, \Phi_1, \Phi_2$ replaced by $E, F_1, F_2$ respectively);
 \item (\emph{Distributional bisolution property}) $\Lambda(PF_1, F_2) = \Lambda(F_1, P F_2) = 0$.
\end{enumerate}

For a state on ${\cal A}$ to be physically relevant, of course, not only must its spacetime smeared two-point function, Equation (\ref{spacetimetwoptfn}), exist, but it must also satisfy the (local or global) Hadamard condition. For general globally hyperbolic spacetimes, we refer to the discussion and references in Footnote \ref{locglob}. In the present paper, the only case we will discuss in detail is the (1+1)-dimensional massless case, the correct formulation of which will, in fact, be the focus of the next section.

\subsection{Case of spacetimes with timelike boundaries}\label{timelikeboundaries}

We would next like to sketch how we expect the quantization procedure for Klein-Gordon fields outlined above could be adapted to the case of `spacetimes with boundary' $(M, g)$, where $M$ is now a manifold with boundary whose boundary is timelike and $(\Int M, g\restriction_{\Int M})$ -- where $\Int M$ denotes the interior of $M$ -- is extendible to a globally hyperbolic spacetime. This class of course includes our Kruskal-in-a-box or Minkowski-with-a-mirror spacetimes. We anticipate that, with more work, all the expectations listed below will be fulfillable for Kruskal.  Of our 1+1 Minkowski-with-mirror spacetime, we will demonstrate in detail in Section 4 that, and how, they are indeed fulfilled so as to have a suitable rigorous treatment of the quantum theory which takes into account the special infra-red properties of this case.

First, we expect that methods akin to those in \cite{hormander2007analysis, gindikin1996mixed} will show that, with the addition of suitable homogeneous boundary conditions on the timelike boundary, the Cauchy problem is well-posed for suitable initial data on suitable initial-value surfaces, as already discussed at the start of Section \ref{basicidea} for the case of Dirichlet boundary conditions. In particular, such suitable initial data, when smooth and of compact support (where it is to be understood that the support could include points on the timelike boundary), should be in one-to-one correspondence with smooth spacelike-compact\footnote{Just as in the globally hyperbolic case, a spacelike-compact function $\phi$ on $M$ is one such that $\supp \phi \subseteq J(K)$ for a compact set $K$, however in this case we allow $K$ to contain points on the timelike boundary.} solutions to this mixed problem, and (once the class of `suitable' initial-value surfaces has been precisely identified) these should in turn be equivalently characterized as being the smooth solutions whose restriction to all suitable initial-value surfaces has compact support. Defining $S$ as the space of spacelike-compact smooth solutions to this mixed problem, we then expect, as discussed in Section \ref{basicidea}, that Equation (\ref{symplformdefinition}) will define a symplectic form $\sigma$ on $S$. 

Furthermore, we expect that one will be able to construct retarded and advanced Green's operators $E^\pm$ which, in addition to satisfying the same requirements as the analogous objects in the globally hyperbolic case -- listed as (i)--(ii) in the previous section -- are such that $E^\pm F\restriction_{\partial M}$ satisfies the given boundary conditions. The domain of $E^\pm$ here should at least include $F\in C_0^\infty(\Int M)$. In the next section we will explicitly construct such objects in the case of the massless wave equation in the region of (1+1)-dimensional Minkowski spacetime to the left of a uniformly accelerating mirror. As we will observe in that case, in general the analogous sequence to (\ref{exactsequence}) will no longer be exact since the kernel of $E=E^--E^+$ will be strictly larger than the image of $P$. Furthermore, both in that case and in the general case one doesn't expect that $E$ will be onto $S$.\footnote{It is an interesting open question (as far as we know) -- again both in the general case and in the (1+1)-dimensional example we will study -- whether the domains of $E^\pm$ can be suitably extended in such a way that the resulting advanced-minus-retarded propagator is onto $S$.}

Assuming that the expectations in the previous paragraphs are fulfilled, we propose that a state on the Weyl algebra $\mathscr{W}(S, \sigma)$ be called Hadamard if its symplectically smeared two-point function exists and if its spacetime smeared two-point function, defined at least on $C_0^\infty(\Int M) \times C_0^\infty(\Int M)$ by Equation (\ref{spacetimetwoptfn}), satisfies the following condition:

\begin{defn}\label{hadamardboundaries}
A bidistribution on $\Int M$ will be said to be globally Hadamard if, for any causally convex open subset ${\cal O}$ of $\Int M$ which, when equipped with the restriction of the metric to $\Int M$, is a globally hyperbolic spacetime in its own right, the restriction of $\Lambda$ to smearings with test functions supported inside ${\cal O}$ is globally Hadamard in the standard sense of \cite{kay1991theorems} mentioned in Footnote \ref{locglob}.  
\end{defn}

Here we recall that a subset $U$ of a spacetime $(N, g)$ is called \emph{causally convex} if, whenever two points $x, y \in U$ can be connected by a causal curve $\gamma$ in $N$, then the portion of $\gamma$ between $x$ and $y$ is entirely contained in $U$. Notice that, if $\mathcal{O}$ is a causally convex globally hyperbolic subset of $\Int M$, then denoting by $E^\pm_\mathcal{O} : C_0^\infty(\mathcal{O}) \to C^\infty(\mathcal{O})$ the unique retarded/advanced Green operators for the Klein-Gordon equation on $\mathcal{O}$, it is easy to verify that, for all $F \in C_0^\infty(\mathcal{O})$, we will have
\begin{equation}\label{Epmrestriction}
[E^\pm F]\restriction_\mathcal{O} = E^\pm_\mathcal{O} F.
\end{equation}
Indeed, that this will be the case follows since, as it is easy to check, $E^\pm$ followed by restriction to $\mathcal{O}$ will have, as an operator on $C_0^\infty(\mathcal{O})$, the support properties and left/right inverse properties which uniquely determine the retarded/advanced Green operators on $\mathcal{O}$.

The above proposal fits nicely with the paradigm of locally covariant (quantum) field theory proposed by Brunetti, Fredenhagen and Verch \cite{brunetti2003generally} and indeed allows an extension of that paradigm to include spacetimes with (timelike) boundaries. Physically, since a spacetime boundary can only be detected by sending a signal to it and receiving one in return, our requirement corresponds to saying that, if we localize the quantum state by only performing measurements within globally hyperbolic regions $\mathcal{O}$ which do not `causally intercommunicate' with the boundary -- i.e.\ such that there are no future-directed piecewise smooth causal curves which begin in $\mathcal{O}$, hit the boundary and then return to $\mathcal{O}$ -- we should not be able to tell whether our universe possesses a real boundary, or whether we are witnessing an `unusual' state on a different, unbounded spacetime. A similar ideology was already contained in \cite{kay1979casimir}, where it was pointed out that such a view is necessary in order to clarify the conceptual issues underlying the Casimir effect. It also appeared in \cite{fewster2007averaged} in the context of the investigation of quantum energy conditions for spacetimes with boundaries.

\section{No-go result for massless fields in 1+1-dimensions with a mirror}\label{1+1}

\subsection{Classical theory}\label{1+1classicaltheory}

In this section we consider in detail the classical theory of a massless real scalar field on the spacetime with boundary, $(M, \eta)$, consisting of the portion $M$ of (1+1)-dimensional Minkowski spacetime `to the left of' (and including) the worldline of a point-like mirror on a timelike trajectory of uniform and eternal acceleration. Without loss of generality we assume that the Minkowskian pseudo-norm of the 2-acceleration is always equal to $-1$.  (Clearly our no-go result does not depend on the numerical value of this quantity.)  Picking a global inertial frame $(t,x)$ such that, when the proper time $\tau$ along the mirror's worldline equals $0$, the mirror is located at $(t=0, x=1)$ and $\d t/\d \tau|_{\tau=0} = 1$, we represent $(M, \eta)$ by $M = \R^2 \setminus \left\{ (t,x) \relmiddle| x^2 - t^2 > 1, \ x>0 \right\}$ and $\eta = \d t^2 - \d x^2$. The manifold $M$ is depicted in Figure \ref{kruskal}, with ($R=1$ and) the vertical (respectively horizontal) axis representing the $t$-axis (respectively $x$-axis).

As already pointed out, this spacetime fails to be globally hyperbolic due to the presence of the timelike boundary given by the mirror's trajectory. It possesses a one-parameter group $\beta_\tau$ of isometries given by the flow of the Killing vector field $k = x\partial/\partial t + t\partial/\partial x$\footnote{Explicitly, in global inertial coordinates, $\beta_\tau(t,x) = (\cosh (\tau) t + \sinh (\tau) x, \sinh (\tau) t + \cosh (\tau) x)$ or, in terms of the null coordinates $(u, v)$ introduced below, $\beta_\tau(u,v) = (e^{-\tau}u, e^{\tau} v)$} describing homogeneous Lorentz boosts in the $x$-direction. $k$ has a bifurcate Killing horizon given by ${\cal H}_A \cup {\cal H}_B$, where ${\cal H}_A = \left\{ (t,x) \relmiddle| t=x \right\}$ and ${\cal H}_B = \left\{ (t,x) \relmiddle| t=-x \right\}$.

We immediately note that any real-valued, smooth solution $\phi$ on $M$ to
\begin{equation}\label{WEBVP}
\square \phi = 0, \qquad \phi\restriction_{\partial M} = 0,
\end{equation}
can be written globally as a sum $\phi(t,x) = f(t-x) + g(t+x)$ for two smooth functions $f$ and $g$ with $g(v) = -f(-1/v)$ for all $v>0$ (cf.\ \cite{kay2015instability}). This can be checked e.g.\ by writing the above equation in the null coordinates $u(t,x)=t-x$ and $v(t,x)=t+x$. It is also easy to check that for any such solution $\phi$ which, in addition, has spacelike-compact support (see Section \ref{quantumkleingordon}), the functions $f$ and $g$ must have the additional property that there exist $u_0$ and $v_0$ such that, for some $a \in \R$, $f(u)=a \ \forall \ u \geq u_0$ and $g(v)= - a\ \forall \ v \leq v_0$. Thus we have complete knowledge of the vector space $S$ of spacelike-compact, smooth (and real-valued) solutions discussed in Section \ref{timelikeboundaries}. And, again as envisaged in that section and in Section \ref{basicidea}, Equation (\ref{symplformdefinition}) defines a manifestly antisymmetric bilinear form $\sigma : S \times S \to \R$, independent of the initial-value surface $\mathscr{C}$ as explained in Section \ref{basicidea}. Since it is easy to check that the Cauchy-Dirichlet problem is well-posed (in the sense of both existence and uniqueness) for initial data of compact support in the interior of the particular initial-value surface $\mathscr{C} = \left\{ (t,x) \relmiddle| t=0 \right\} \cap M$, one could prove the non-degeneracy of $\sigma$ directly by picking, for any $\phi_1 \in S$, which will have some initial data $(\varphi_1, \pi_1) \in C_0^\infty(\mathscr{C}) \oplus C_0^\infty(\mathscr{C}),$\footnote{Note that, since $\mathscr{C}$ is a manifold with boundary, functions in $C_0^\infty(\mathscr{C})$ -- which are by definition smooth functions with conpact support on $\mathscr{C}$ -- need not be supported away from the boundary; indeed, they needn't even vanish \emph{at} the boundary (although for this specific choice of $\mathscr{C}$, both pieces of Cauchy data will have to vanish at the boundary because of the Dirichlet boundary condition).} $\phi_2$ to be the solution with initial data $(\varphi_2, \pi_2) \in C_0^\infty(\Int \mathscr{C}) \oplus C_0^\infty(\Int \mathscr{C})$ where $(\varphi_2, \pi_2)$ approximate $(-\pi_1, \varphi_1)$ (respectively) `sufficiently well' for $\sigma(\phi_1, \phi_2)$ to be greater than $0$. This can always be done by picking $\varphi_2 = - \psi \pi_1$ and $\pi_2 = \psi \varphi_1$ where $\psi \in C_0^\infty(\Int \mathscr{C}) \subset C_0^\infty(\mathscr{C})$ is such that $0\le\psi\le 1$ and $\psi=1$ everywhere but on a small enough neighbourhood of the boundary point $(t=0, x=1)$ of $\mathscr{C}$. Indeed, we expect a generalization of this strategy to apply to the more general setup described in Section \ref{timelikeboundaries}. We will also provide another, independent, proof of the non-degeneracy of $\sigma$ later in this section.

Thus we have endowed $S$ with the structure of a symplectic vector space $(S, \sigma)$. A simple calculation, which e.g.\ starts with the expression for $\sigma$ in terms of the $t=0$ initial-value surface mentioned above and then involves a change of variables, shows that, for any $\phi_1, \phi_2 \in S$,
\begin{align}
 \sigma(\phi_1, \phi_2)
&= 2\int_{-\infty}^{+\infty} f_1 (u) f_2'(u) \, \d u + 2\int_{-\infty}^{0} g_1 (v) g_2'(v) \, \d v \label{symplformB} \\
&= 2\int_{0}^{+\infty} f_1 (u) f_2'(u) \, \d u + 2\int_{-\infty}^{+\infty} g_1 (v) g_2'(v) \, \d v, \label{symplformA}
\end{align}
where, $f_1, g_1, f_2, g_2$ are any smooth functions such that $\phi_1 (t,x) = f_1(t-x) + g_1(t+x)$ and $\phi_2(t,x) = f_2(t-x) + g_2(t+x)$. These explicit expressions will be important in the next paragraph.

Let $S_A$ and $S_B$ denote the linear subspaces of $S$ consisting of those solutions which `fall entirely through' ${\cal H}_A$ and ${\cal H}_B$ respectively. A geometric definition of these was already given in the third paragraph of Section \ref{basicidea}. However, a more explicit characterization is also available here: $\phi \in S_B$ (respectively $\phi \in S_A$) if and only if $\phi(t,x) = f(t-x) + g(t+x)$ with the `right mover' $f$ belonging to $C^\infty_0(\R)$ and the `left mover', $g(v)$, being equal to zero for all $v \leq 0$, and to $-f(-1/v)$ for all $v>0$ (respectively the `left mover' $g$ belonging to $C^\infty_0(\R)$ and the `right mover', $f(u)$, being equal to zero for all $u \geq 0$, and to $-g(-1/u)$ for all $u<0$). Thus, solutions in $S_B$ (respectively $S_A$) are uniquely determined by their restriction to ${\cal H}_B$ (respectively ${\cal H}_A$). And indeed, the initial value problem is well-posed on Cauchy surfaces which include portions of ${\cal H}_B$ (respectively ${\cal H}_A$), for data supported on those portions. For any pair $\phi_1, \phi_2$ of $S_B$-solutions (respectively $S_A$-solutions), the second (respectively first) summand on the right-hand side of Equation (\ref{symplformB}) (respectively Equation (\ref{symplformA})) vanishes, and thus $\sigma(\phi_1, \phi_2)$ can be interpreted as twice the integral along ${\cal H}_B$ (respectively ${\cal H}_A$) of $\phi_1 \partial_u \phi_2\, du$ (respectively $\phi_1 \partial_v \phi_2\, dv$). Moreover, let $(S_\mathbb{M}, \sigma_\mathbb{M})$ denote the symplectic vector space of spacelike-compact, smooth, real-valued solutions to the massless wave equation on $(\R^2, \eta)$, and let $S_{\mathrm{r\text{-}mov}}$ and  $S_{\mathrm{l\text{-}mov}}$ denote the vector subspaces of $S_\mathbb{M}$ consisting of right-moving and left moving (respectively) solutions.  (That is solutions which, respectively, are functions of $u$ only and of $v$ only.) Then, as is well known (or easy to show), 
$(S_{\mathrm{r\text{-}mov}}, \sigma_\mathbb{M})$ and $(S_{\mathrm{l\text{-}mov}},\sigma_\mathbb{M})$\footnote{Throughout the text we adopt the convention that, if $(S,\sigma)$ is a symplectic vector space and $T$ is a vector subspace of $S$, then the presymplectic vector space $(T, \sigma\restriction_{T \times T})$ is written simply as $(T, \sigma)$.} are symplectic vector spaces in their own right and one has the following important result, whose proof is immediate.
\begin{prop}\label{symplectomorphism}
The map $T_B : S_B \to S_{\mathrm{r\text{-}mov}}$, defined by sending $\phi \in S_B$ to the unique Minkowski-space right-moving solution with the same data as $\phi$ on ${\cal H}_B$, is a presymplectic isomorphism between $(S_B, \sigma)$ and $(S_{\mathrm{r\text{-}mov}}, \sigma_\mathbb{M})$. Thus in particular $(S_B, \sigma)$ is a symplectic space and the map is a symplectic isomorphism. (And similarly, with $B$ replaced by $A$ and $\mathrm{r\text{-}mov}$ replaced by $\mathrm{l\text{-}mov}$.)
\end{prop}

We can now also define a proper linear subspace $S_0$ of $S$ by $S_0 := S_A + S_B$, and subspaces $S_A^R, S_A^L \subset S_A$, $S_B^R, S_B^L \subset S_B$ just as explained in Section \ref{basicidea}, that is 
\begin{equation*}
S_A^{L/R} := \left\{\phi \in S_A \relmiddle| \supp (\phi\restriction_{{\cal H}_A}) \subset {\cal H}_A^{L/R} \right\},
\end{equation*}
with ${\cal H}_A^L$ and ${\cal H}_A^R$ the `left' and `right' portions of the $A$-horizon, i.e.\ ${\cal H}_A^L := {\cal H}_A \cap \left\{ (t,x) \relmiddle| x < 0 \right\}$ and ${\cal H}_A^R := {\cal H}_A \cap \left\{ (t,x) \relmiddle| x > 0 \right\}$ (and similarly with $S_B^{L/R}$ and ${\cal H}_B^{L/R}$). It is clear that $(S_A, \sigma)$, $(S_B, \sigma)$, $(S_A^{L/R}, \sigma)$, $(S_B^{L/R}, \sigma)$ are all symplectic spaces (indeed, for $(S_A,\sigma)$,  $(S_B, \sigma)$ this was already established in Proposition \ref{symplectomorphism}). It is also clear that $T_B$ restricts to a symplectic isomorphism between $(S_B^{L/R}, \sigma)$ and $(S_{\mathrm{r\text{-}mov}}^{L/R},\sigma_{\mathbb{M}})$, where $(S_{\mathrm{r\text{-}mov}}^{L/R},\sigma_{\mathbb{M}})$ is the symplectic subspace of $(S_{\mathrm{r\text{-}mov}},\sigma_{\mathbb{M}})$ consisting of purely right-moving solutions in $S_\mathbb{M}$ whose data on ${\cal H}_B$ is supported strictly to the left/right of the origin (and similarly, with $B$ replaced by $A$ and $\mathrm{r\text{-}mov}$ replaced by $\mathrm{l\text{-}mov}$). 

We wish next to show that the presymplectic space $(S_0,\sigma)$ is also actually a symplectic space.\footnote{\label{gap}While we were proving Proposition \ref{symplecticity} we noticed that there seems to be a gap in the arguments on a corresponding issue in \cite{kay1991theorems}: While it was clear that the $(S_A, \sigma)$ and $(S_B,\sigma)$ of that paper are symplectic spaces (and the same is also true for the spaces called $(\tilde{S}_A, \hat{\sigma})$ and $(\tilde{S}_B, \hat{\sigma})$, as we show in Appendix \ref{appB}) it was also tacitly assumed that (with our fiction) the space called $(S_0,\sigma)$ and (without our fiction) the space called $(\tilde{S}_0, \hat{\sigma})$ are symplectic spaces. However this was never established there. We describe possible ways of filling this gap in some cases of physical interest in Appendix \ref{appB} here.} In fact we will prove a stronger result.  Note first that the formula on the right-hand side of Equation (\ref{symplformdefinition}) is still well-defined and antisymmetric when only one of the solutions is spacelike-compact, and Equations (\ref{symplformB})--(\ref{symplformA}) are also still valid in that case.

\begin{prop}\label{symplecticity}
Suppose $\phi$ is any (not necessarily spacelike-compact) smooth solution to (\ref{WEBVP}) on $M$ which is symplectically orthogonal to both $S_A$ and $S_B$, i.e.\ $\sigma(\phi_A, \phi) = 0 = \sigma(\phi_B, \phi)$ for all $\phi_A \in S_A$ and $\phi_B \in S_B$. Then $\phi = 0$.
\end{prop}
\begin{proof}
Let $f, g$ be smooth functions such that $\phi(t,x) = f(t-x)+g(t+x)$. Solutions in $S_B$ have the form $\phi_B (t,x) = h(t-x) + k(t+x)$ where $h$ is any function in $C_0^\infty(\R)$ and $k(v) =  - \vartheta(v) h(-1/v)$. Therefore, if $\phi$ is symplectically orthogonal to $S_B$ then Equation (\ref{symplformB}) implies that
\begin{equation*}
    \int_{-\infty}^{+ \infty} h(u) f'(u) \, \d u = 0
\end{equation*}
for \emph{all} $h \in C_0^\infty(\R)$. This implies that $f'$ is identically zero and thus that $f$ equals a constant. A similar argument shows that $g$ equals a constant. Thus $\phi$ is also constant. But then it must be zero since it is assumed to vanish on $\partial M$.
\end{proof}

As already anticipated in the Introduction, two further important observations for the purposes of this paper are that, with the above definitions and using Equations (\ref{symplformB})--(\ref{symplformA}), it is clearly the case that
\begin{itemize}
\item $S_A^R=S_B^R$,
\item $S_B^L$ is symplectically orthogonal to $S_A$. Similarly, $S_A^L$ is symplectically orthogonal to $S_B$.
\end{itemize}

As final `classical' ingredients necessary to formulate and then to prove our no-go result in the remainder of this section, we need to construct retarded/advanced Green operators $E^\pm$ appropriate to our Cauchy-Dirichlet problem on $M$, as discussed in Section \ref{timelikeboundaries}. Namely, $E^\pm$ should be such that, for all $F \in C_0^\infty(\Int M)$,
\begin{align}
\square E^\pm F = E^\pm \square F &= F, \label{advancedretarded1}\\
E^\pm F \restriction_{\partial M} &= 0, \label{advancedretarded2}\\ 
\supp (E^\pm F) &\subseteq J^\pm (\supp F). \label{advancedretarded3}
\end{align}
The resulting causal propagator $E = E^- - E^+ : C_0^\infty(\Int M) \to C^\infty(M)$ will then clearly map to $S$. 

We will now argue that $E^\pm$ with the above properties can indeed be constructed. In what follows, for each $p \in M$ we denote by $m_+(p)$ (respectively $m_-(p)$) the set of all future (respectively past) endpoints on $\partial M$ of (smoothly) inextendible null geodesics passing through $p$. Equivalently, $m_\pm(p)$ is the intersection between $\partial M$ and the topological boundary (in $M$) of $J^\pm(p) = J^\pm(\{p\})$. In particular, $m_\pm(p)$ is either empty or a singleton, and $m_{\pm}(p) = \{ p \}$ if $p \in \partial M$. See Figure \ref{Eplus}.

It is well-known and easy to verify that the unique advanced and retarded Green operators for the scalar wave equation in full (1+1)-dimensional Minkowski spacetime, which we denote by $E_\mathbb{M}^\pm$, are given by
\begin{equation*}
[E_\mathbb{M}^\pm F](p) = \frac{1}{2} \int_{J_\mathbb{M}^\mp (p)} F \, \d \mu_\eta,
\end{equation*} 
where $p \in \R^2$, $F \in C_0^\infty(\R^2)$, $J_\mathbb{M}^\mp (p)$ is the causal past/future of $p$ in the full Minkowski space, and $\d \mu_\eta$ denotes the metric volume element. Consequently, the causal propagator $E_\mathbb{M}$ is given by
\begin{equation}\label{CausProp}
[E_\mathbb{M} F](p) = \frac{1}{2} \Bigg\{ \int_{J_\mathbb{M}^+ (p)} - \int_{J_\mathbb{M}^- (p)} \Bigg\} F \, \d \mu_\eta = \frac{1}{2} \Bigg\{ \int_{V(p)} - \int_{U(p)} \Bigg\} F \, \d \mu_\eta
\end{equation}
where we have defined the sets $V(p) :=\{ p' : v(p') \geq v(p) \}$, $U(p) :=\{ p' : u(p') \leq u(p) \}$, with $u$ and $v$ the global null coordinates defined above. The first term in the rightmost expression is a function of the $v$-coordinate of $p$ only, while the second is a function of the $u$-coordinate only. Thus one retrieves the expression of the solution as a sum of a left mover and a right mover, which we denote by $g_\mathbb{M}(v)$ and $f_\mathbb{M}(u)$ respectively.

We next make a definition before finally being able to state the result on existence of advanced and retarded Green operators in the presence of our mirror.

\begin{defn}\label{C_00}
For any open subset $X \subset \R^2$, we denote the space of compactly supported smooth functions on $X$ with vanishing integral with respect to the Minkowski metric measure by $C_{00}^\infty(X)$. That is,
\begin{equation*}
C_{00}^\infty(X) := \left\{ F \in C_0^\infty(X) \relmiddle| \int_X F \, \d \mu_\eta = 0 \right\}.
\end{equation*}
\end{defn}
Note that in what follows we will sometimes identify test functions defined on an open subset $X$ with test functions on the whole of Minkowski space (by extending them to be zero outside of $X$). It is easy to see from Equation (\ref{CausProp}) that, in the full Minkowski space theory, $E_\mathbb{M} [C_{00}^\infty(\R^2)]$ consists of all solutions (to the massless wave equation) of the form $f(t-x) + g(t+x)$ with $f, g \in C_0^\infty(\R)$. That is, defining the subspaces $S_{A, \mathbb{M}}$, $S_{B, \mathbb{M}}$ and $S_{0, \mathbb{M}} := S_{A, \mathbb{M}} + S_{B, \mathbb{M}}$ of $S_\mathbb{M}$, in a manner analogous to the way we defined $S_A$, $S_B$ and $S_0 = S_A + S_B$, one has $S_{0, \mathbb{M}} = E_\mathbb{M}[C_{00}^\infty(\R^2)]$.

\begin{figure}
   \centering
    \includegraphics[scale=0.5]{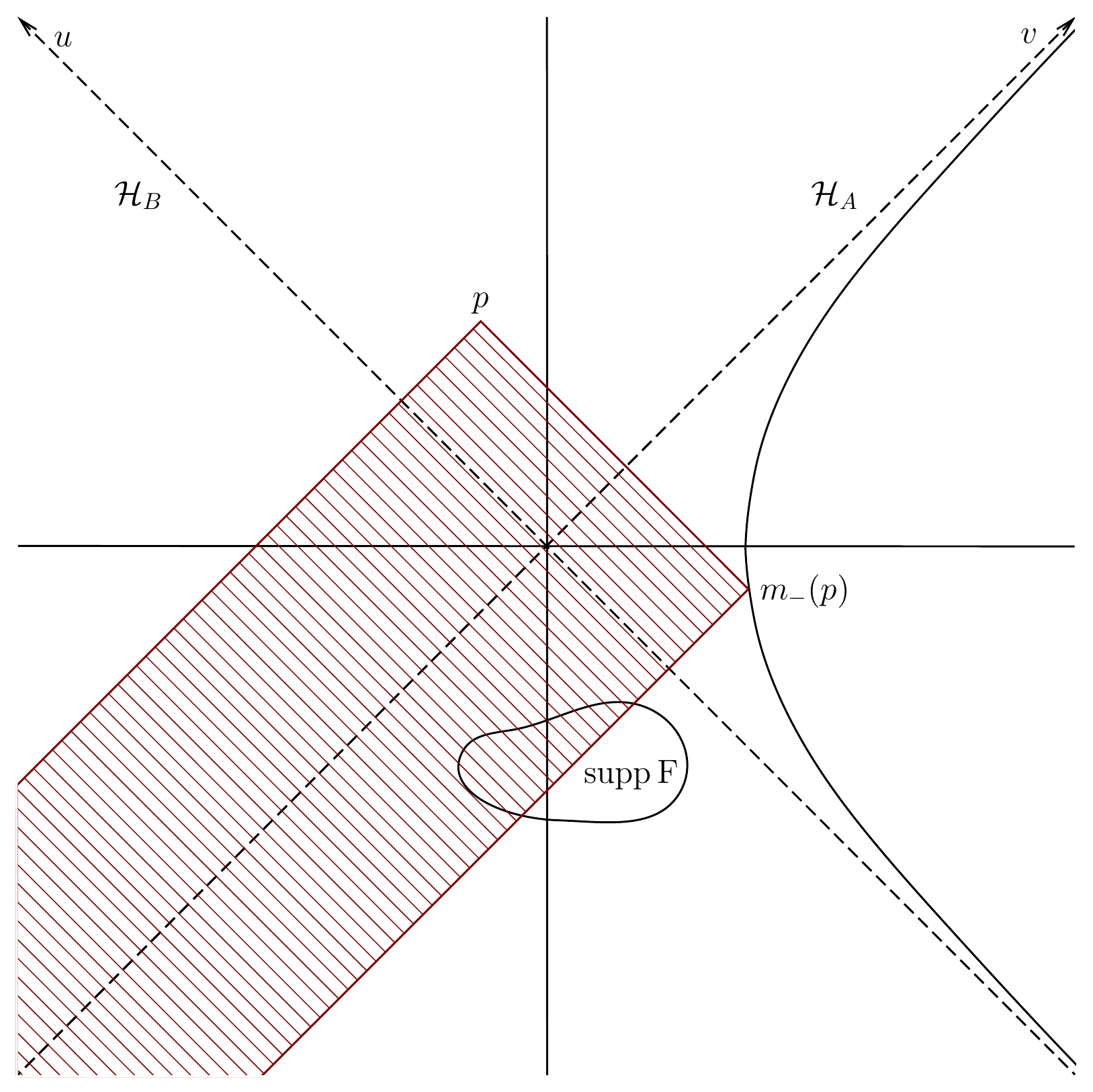}
   \caption{Illustration of the definition of the retarded propagator $E^+$ given in Theorem \ref{greensoperators}. Integrating one half times the source $F$ over the shaded region gives $[E^+F](p)$. The definition of $E^-$ can, of course, be illustrated similarly. \label{Eplus}}
 \end{figure}

\begin{thm}\label{greensoperators}
The linear operators $E^\pm : C_0^\infty(\Int M) \to C^\infty(M)$ defined, for all $F \in C_0^\infty(\Int M)$ and $p \in M$, by
\begin{align}\label{advancedretardeddefn}
[E^\pm F] (p) &= \frac{1}{2} \left\{ \int_{J^\mp(p)} - \int_{J^\mp(m_\mp (p))} \right\} F \, \d \mu_\eta
 \nonumber \\
&= [E_\mathbb{M}^\pm F](p) - \frac{1}{2}\int_{J^\mp(m_\mp (p))} F \, \d \mu_\eta,
\end{align}
(see Figure \ref{Eplus}) satisfy Equations (\ref{advancedretarded1})--(\ref{advancedretarded3}).  Furthermore, $S_0 = E[C_{00}^\infty(\Int M)]$.
\end{thm}
\begin{rk}
	The second summand on the right-hand side of Equation (\ref{advancedretardeddefn}) equals zero (for any test function) at any point $p$ for which $m_\mp(p) = \emptyset$ (whereupon the integration domain $J^\mp(m_\mp(p))$ is also empty). When $m_\mp(p)$ consists of the point $p_\mp$, it equals $[E^\pm_\mathbb{M} F](p_\mp)$.
\end{rk}
\begin{proof}
That each $E^\pm F$ is smooth is obvious since our test functions have compact support. The boundary condition, Equation (\ref{advancedretarded2}), and the support property, Equation (\ref{advancedretarded3}), also hold trivially.

We now turn to the equations in (\ref{advancedretarded1}), i.e.\ to the two-sided inverse property, on the domain $C_0^\infty(\Int M)$, of $E^\pm$ with respect to the d'Alembert operator $\square$. We carry out the proof explicitly in the case of $E^+$; the arguments for $E^-$ are analogous.  In view of the fact that the corresponding object $E^+_\mathbb{M}$ on full Minkowski space is already known to satisfy the analogous two-sided inverse property for all test functions (and thus in particular for those supported in $\Int M$), we need to check that the operator $D^+ = E^+ - E^+_\mathbb{M}$, whose action is defined by the second summand on the right-hand side of Equation (\ref{advancedretardeddefn}), is such that $D^+ \square F = 0 = \square D^+ F$ whenever $F \in C_0^\infty(\Int M)$. Using the remark above and, again, the left-inverse property for $E^+_\mathbb{M}$, it is easy to see that the first of these identities holds because any $F \in C_0^\infty(\Int M)$ vanishes on $\partial M$. To verify the second identity, we first express $D^+ F$ in terms of the null coordinates $u$ and $v$. For any $p \in M$, $m_-(p)$ is empty if $v(p) \leq 0$, and contains only the point with null coordinates $u_- = -1/v(p)$ and $v_- = v(p)$ if $v(p) > 0$. Therefore, one has
\begin{equation}\label{retardednullcoords} [\widetilde{D^+ F}](u,v) = - \frac{\vartheta(v)}{4} \int_{u' \leq -1/v} \tilde{F}(u',v') \, \d u' \, \d v', 
\end{equation}
where the tilde indicates that one is dealing with the coordinate expression of a function in the $(u, v)$ coordinate system, and $\vartheta$ denotes the Heaviside step function. The right-hand side of Equation (\ref{retardednullcoords}) is clearly annihilated by $\partial/\partial u$, and thus in particular by $\square = 4\partial^2/\partial u \partial v$. This completes the proof of the right-inverse property for $E^+$.

In order to prove the second statement in the theorem, we first point out that it is straightforward to check that, for any $F \in C_0^\infty(\Int M)$,
\begin{equation}\label{causalpropagatorimages}
[\widetilde{EF}](u,v) = f_\mathbb{M}(u) + g_\mathbb{M}(v) - \vartheta(-u) g_\mathbb{M}(-1/u) - \vartheta(v) f_\mathbb{M}(-1/v),
\end{equation}
where $f_\mathbb{M}$ and $g_\mathbb{M}$ denote the right- and left-moving parts of $E_\mathbb{M}F$ obtained in the manner described in the discussion under Equation (\ref{CausProp}). That is,
\begin{equation}\label{rightleftmovingparts}
f_\mathbb{M}(u) = - \frac{1}{4} \int_{u' \leq u} \tilde{F}(u',v') \, \d u' \, \d v' \quad \text{and} \quad g_\mathbb{M}(v) = \frac{1}{4} \int_{v' \geq v} \tilde{F}(u',v') \, \d u' \, \d v'.
\end{equation}
Since $f_\mathbb{M}$ and $g_\mathbb{M}$ have compact support when $F \in C_{00}^\infty(\Int M)$, it follows that $E[C_{00}^\infty(\Int M)] \subseteq S_0$. To prove the reverse inclusion, it clearly suffices to show that $S_A$ and $S_B$ are individually contained in $E[C_{00}^\infty(\Int M)]$. We give the argument for $S_B$, the one for $S_A$ being entirely analogous.
If $\phi \in S_B$ then $\tilde{\phi}(u,v) = h(u) + k(v)$ where $h \in C_0^\infty(\R)$ and $k(v) = - \vartheta(v) h(-1/v)$. In view of Equation (\ref{causalpropagatorimages}), it therefore suffices to find an $F \in C_{00}^\infty(\Int M)$ such that $f_\mathbb{M}$ and $g_\mathbb{M}$ in Equation (\ref{rightleftmovingparts}) equal $h$ and $0$ respectively (i.e.\ $F$ needs to integrate to zero, be supported in $\Int M$ and generate the pure right-mover -- in the full Minkowski space theory -- described by $h$). This can be done as follows: Pick any $\chi \in C_0^\infty(\R)$ with the properties that $\supp \chi \subset (-\infty,0)$ and $\int_{\R} \chi(x) \, \d x = 1$. Then, the function $F$ defined by 
\begin{equation}
\tilde{F}(u,v) = -4 h'(u)\chi(v)
\end{equation}
clearly fulfills the required properties.
\end{proof}

To make contact with the general discussion in Section \ref{timelikeboundaries}, we remark that we have \emph{not} proved that $E : C_0^\infty(\Int M) \to C^\infty(M)$ is onto $S$. Indeed, as pointed out there, we don't expect this to be the case. Nor, as also anticipated there, is the kernel of the causal propagator constructed in Theorem \ref{greensoperators} equal to $\square[C_0^\infty(\Int M)]$, as one can see from Equation (\ref{causalpropagatorimages}). Indeed, one need only pick a test function $F \in C_0^\infty(\Int M)$ which, on the entire Minkowski space, would propagate to a non-zero solution with right- and left-moving parts $f_\mathbb{M}$ and $g_\mathbb{M}$ respectively (obtained again in the manner described in the discussion under Equation (\ref{CausProp})), which are such that $f_\mathbb{M}(u) = \vartheta(-u)g_\mathbb{M}(-1/u)$ and $g_\mathbb{M}(v) = \vartheta(v) f_\mathbb{M}(-1/v)$ for all $u, v \in \R$.   Then $EF=0$ but $F$ cannot equal $\square G$ for any $G\in C_0^\infty(\Int M)$ since $E_\mathbb{M}F\ne 0$ in full Minkowski space.  See Figure \ref{kernel}.

\begin{figure}
   \centering
    \includegraphics[scale=0.5]{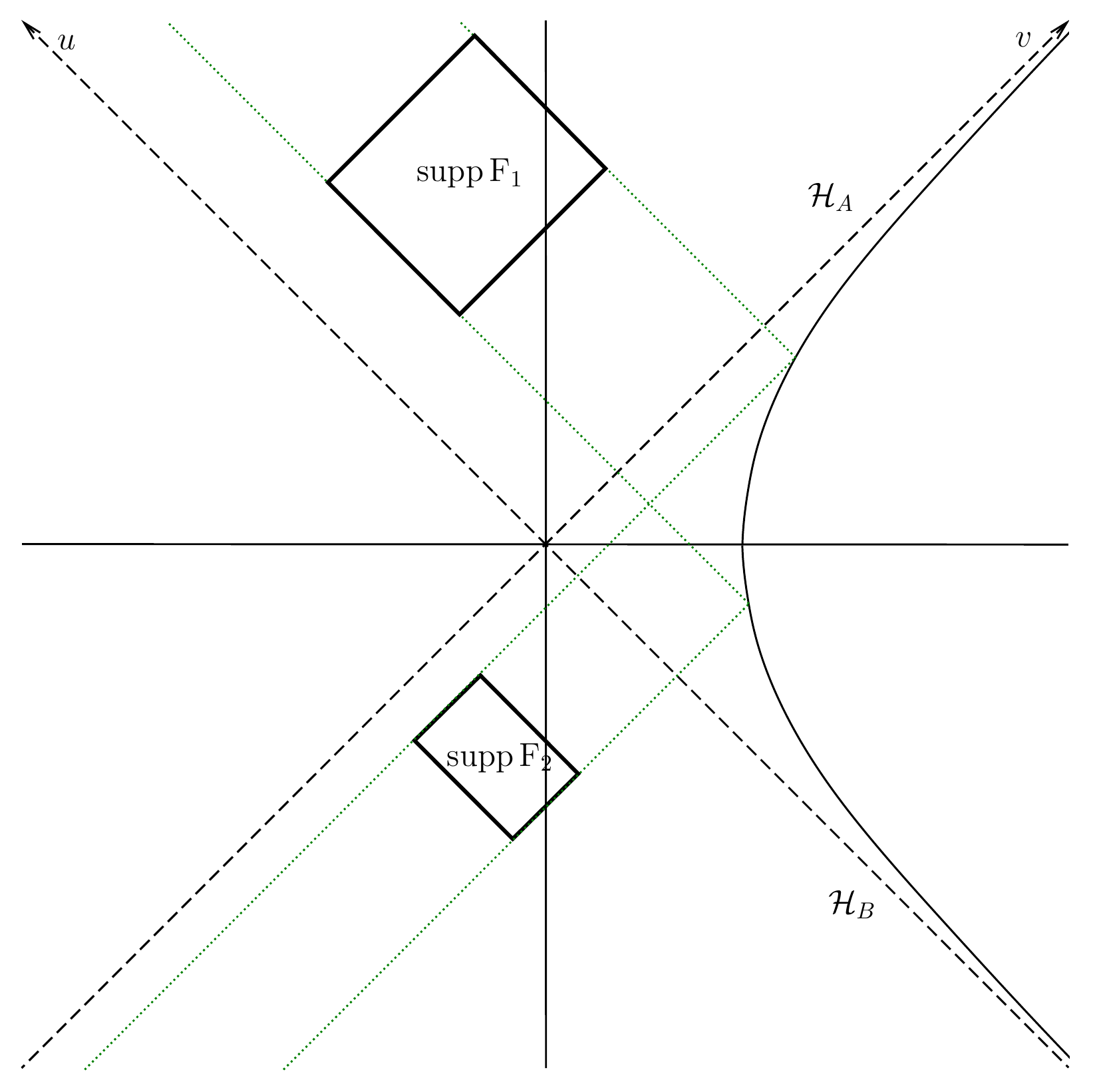}
   \caption{Illustration of the failure of $\ker E$ to be equal to $\square [C_0^\infty(\Int M)]$. The $F$ in the discussion in the main text (which in this illustration has disconnected support) is to be identified with $F_1-F_2$. Then $E F = 0$ but $F$ cannot equal $\square G$ for any $G \in C_0^\infty(\Int M)$ since $E_\mathbb{M}F \neq 0$ in full Minkowski space. \label{kernel}}
 \end{figure}

As another side remark, we note that, equipped with the above results, one can straightforwardly imitate an argument which is standard in the globally hyperbolic setup (see e.g.\ \cite[Lemma 3.2.2]{baer2007wave}, or after Equation (3.18) in \cite{kay1991theorems}) to show that, for any $F \in C_0^\infty(\Int M)$ and $\phi \in S$,
\begin{equation}\label{usefulidentity}
\int_M F \phi \, \d \mu_\eta = \sigma(E F, \phi).
\end{equation}
And Equation (\ref{usefulidentity}) provides the alternative way, promised above, to show the non-degeneracy of $\sigma$. Indeed, for any given $\phi \in S$ it is clearly possible to find a test function $F \in C_0^\infty(\Int M)$, not in the kernel of $E$ (i.e.\ not generating the zero solution) and such that $\int_M F \phi \, \d \mu_\eta \neq 0$ -- any $F$ which is everywhere non-zero and is sufficiently localized around a point where $\phi$ attains a non-zero value will do.

We conclude this section by briefly discussing the action of Lorentz boost isometries on elements of $S$. The one-parameter group $(\beta_\tau)_{\tau \in \R}$ of Lorentz-boost isometries yields a one-parameter abelian group of linear symplectomorphisms $\mathscr{T}_\tau : S \to S$ via pullback by the inverse maps, i.e.\ $\mathscr{T}_\tau \phi := \phi \circ \beta_{-\tau}$. Explicitly, if $\phi(t,x) = f(t-x)+g(t+x)$ then $[\mathscr{T}_\tau \phi] (t,x) = f_\tau(t-x) + g_\tau(t+x)$ where $f_\tau(u) = f(e^{a \tau}u)$ and $g_\tau(v) = g(e^{-a \tau}v)$.

\subsection{The infra-red pathology and the Hadamard notion}\label{infra-red}

We now wish to discuss the prospects for identifying an appropriate framework for the quantization of the massless field on $(M, \eta)$. We first recall some of the issues arising in the quantization of massless fields in \emph{full} (1+1)-dimensional Minkowski spacetime.

As we mentioned in the Introduction and in Section \ref{basicidea}, in attempting to define a ground state representation there, one is faced with an infra-red pathology (see e.g.\ \cite{schroer1963infrateilchen, wightman1967introduction, streater1970fermion, kay1985double, fulling1987temperature, derezinski2006quantum}). To recall the issue: One might attempt to define the quantum field as a genuine operator-valued distribution\footnote{It is irrelevant to this discussion whether the quantum field is to be smeared with test functions in $C_0^\infty(\R^2)$ or, say, test functions in Schwartz space $\mathscr{S}(\R^2; \R)$. But we will work with the former space because it's technically more appropriate for our needs in this section.} by proceeding in the usual way involving creation and annihilation operators on the standard bosonic Fock space $\mathscr{F} = \bigoplus_{n=0}^\infty L^2(\R)^{\odot n}$. One would then demand that the Fock vacuum vector $\Omega$ belong to a common invariant (and dense) domain for all thus defined field operators. However, in general the resulting one-particle vectors $\hat{\phi} (F) \Omega$ -- generated by acting on the vacuum with the candidate quantum field smeared with an arbitrary test function $F$ on spacetime -- might not be square integrable. In fact, if $\tilde{F}(k)$ is the Fourier transform, $(1/\sqrt{2\pi})\int_{\R^2} F(x)e^{-ik\cdot x} \, \d^2 x$ of $F$, the vacuum belongs to the domain of $\hat{\phi}(F)$ if and only if $\tilde{F}(0) = 0$. This problem starkly manifests itself at the level of the tentative `two-point function', which is formally given by 
\begin{equation}\label{twopointfunctionfail}
\innerprod{\Omega}{\hat{\phi}(F) \hat{\phi}(G) \Omega} = \pi\int_{-\infty}^{+\infty} \tilde{F}(|p|,-p) \tilde{G}(-|p|,p) \, \frac{\d p}{|p|}.\end{equation}
Indeed, the above clearly diverges (logarithmically) unless one of $\tilde{F}(0)$ or $\tilde{G}(0)$ equals zero.

Thus the usual quantization procedure fails to produce, via Equation (\ref{twopointfunctionfail}), a bidistribution, $\Lambda$, on $\R^2$ representing two-point correlators, because one can't allow for generic test functions. If, however, one restricts to smearings with elements of the linear subspace $C_{00}^\infty(\R^2)$ of Definition \ref{C_00}, then both this `two-point functional' exists and (by construction via creation and annihilation operators) satisfies the positivity properties $\Lambda(F, F) \geq 0$, $E_{\mathbb{M}}(F,G)^2 \leq 4 \Lambda(F,F) \Lambda(G,G)$\footnote{If we let $\mathscr{D}_0(\R^2)$ denote the complexification of $C_{00}^\infty(\R^2)$ then these positivity conditions can be succinctly expressed as $\Lambda^\mathbb{C}(\bar F, F) \geq 0 \ \forall \ F \in \mathscr{D}_0(\R^2)$, where $\Lambda^\mathbb{C}$ denotes the extension by complex bilinearity of $\Lambda$ to a bilinear form on $\mathscr{D}_0(\R^2)$.} required for a probabilistic interpretation.

In the Weyl-algebraic approach to quantization which we adopt in this paper (see Appendix \ref{appA}), what is problematic is the attempt to define a ground state with respect to time translations on the Weyl algebra ${\cal A}_\mathbb{M} = \mathscr{W}(S_\mathbb{M}, \sigma_\mathbb{M})$ generated by the symplectic space $(S_\mathbb{M}, \sigma_\mathbb{M})$ defined in Section \ref{1+1classicaltheory}. But we observe that, if we restrict to the Weyl subalgebra ${\cal A}_{0, \mathbb{M}} = \mathscr{W}(S_{0, \mathbb{M}} = E_\mathbb{M} [C_{00}^\infty(\R^2)], \sigma_\mathbb{M})$ then there \emph{is} an unproblematic ground state with respect to time translations, namely the state whose spacetime smeared two-point function is precisely the `two-point functional' of the previous paragraph. In Section \ref{theorem} we will refer to this state on ${\cal A}_{0, \mathbb{M}}$ -- which, we remark in passing, is a quasi-free state -- as $\omega_\mathbb{M}$, and to its symplectically smeared two-point function as $\lambda_\mathbb{M}$. In view of this, from now on we adopt the view (essentially what in \cite{fulling1987temperature} is termed the `liberal' approach to dealing with the infra-red pathology) that our `physical algebra' is this Weyl subalgebra ${\cal A}_{0, \mathbb{M}}$ and `physical states' are to be sought amongst positive linear functionals on ${\cal A}_{0, \mathbb{M}}$.

A price to pay for working in this framework is that the spacetime smeared two-point functions of our thus-defined physical states are only defined as bilinear functionals $C_{00}^\infty(\R^2) \times C_{00}^\infty(\R^2) \to \C$, and therefore do not define true bidistributions on $\R^2$. As a result, what one might mean by a globally (or even locally!) `Hadamard' state becomes problematic. We propose to overcome this by declaring that a state on ${\cal A}_{0, \mathbb{M}}$ be called globally Hadamard if its spacetime smeared two-point function $\Lambda : C_{00}^\infty(\R^2) \times C_{00}^\infty(\R^2) \to \mathbb{C}$ (exists and) admits an extension $\Lambda^\mathrm{ext} : C_0^\infty(\R^2) \times C_0^\infty(\R^2) \to \mathbb{C}$ which is globally Hadamard (on $\R^2$). Note that this extension need not satisfy any positivity property beyond positivity (in the above sense) when restricted to smearings in $C_{00}^\infty(\R^2)$. The (1+1)-dimensional version of the global Hadamard condition for bidistributions was written down in \cite{moretti2003comments} (along with versions appropriate to all other spacetime dimensions). For a massless theory in any globally hyperbolic open subset ${\cal O}$ of (1+1)-dimensional Minkowski space, it simply amounts to the following.

\begin{defn}[\emph{Global Hadamard condition on ${\cal O}$, massless case}]\label{hadamard1+1}
A bidistribution $\Lambda$ on ${\cal O}$ satisfies the global Hadamard condition if there exists a Cauchy surface $\mathscr{C}$ for $({\cal O}, \eta)$, a causal normal neighbourhood ${\cal N} \subseteq {\cal O}$ of $\mathscr{C}$, a `smoothing function' $\chi \in C^\infty({\cal N} \times {\cal N})$, a global temporal function $T$ on ${\cal O}$ increasing towards the future,\footnote{We refer to \cite{kay1991theorems, radzikowski1996microlocal} for precise characterizations of $\mathcal{N}$, $\chi$ and $T$.} and a \emph{smooth} function $H_{\cal N}$ on $\mathcal{N} \times \mathcal{N}$ such that, for all $F, G \in C_0^\infty({\cal N})$,
\begin{equation*}
\Lambda(F,G)  = \lim_{\varepsilon \to 0^+} \int_{{\cal N} \times \cal{N}} \left( - \frac{\chi(x,y)}{4 \pi} \ln\frac{- s_{\varepsilon, T}(x,y)}{\lambda^2} + H_{\cal N}(x,y) \right) F(x) G(y) \, \d \mu_\eta(x) \, \d \mu_\eta (y).
\end{equation*}
In the above, for all $\varepsilon >0$,
\begin{equation*}
s_{\varepsilon, T}(x,y) := s(x,y) - 2 i \varepsilon(T(x) - T(y)) - \varepsilon^2,
\end{equation*}
with $s(x,y) = (x-y)^2$ and the branch-cut of the logarithm chosen to lie on the negative real axis. Finally, $\lambda$ is a length scale introduced for dimensional reasons, but clearly the property being defined does not depend on it.
\end{defn}

Clearly, the ground state on the physical algebra ${\cal A}_{0, \mathbb{M}}$ is a globally Hadamard state in this sense. To prepare the ground for our discussion in the case of the spacetime $(M, \eta)$ we're interested in, where the Lorentz boosts are the only continuous isometries, we notice that actually more is true about this state on ${\cal A}_{0, \mathbb{M}}$, namely that one can find an extension of its spacetime smeared two-point function which, on its larger domain $C_0^\infty(\R^2) \times C_0^\infty(\R^2)$, is still boost-invariant, a weak bisolution of the wave equation, and satisfies the canonical commutation relations. Indeed, $\Lambda_\mathbb{M}$ defined by
\begin{equation}\label{minktwopointfunction} \Lambda_\mathbb{M}(F,G) = -\frac{1}{4\pi} \lim_{\varepsilon \to 0^+} \int \log\left[\frac{-(x-y)^2 + i \varepsilon (x^0 - y^0)}{\lambda^2}\right] F(x) G(y) \, \d \mu_\eta(x) \, \d \mu_\eta(y) \end{equation}
gives such an extension for any choice of length scale $\lambda$. It can be seen that, indeed, no such extension can satisfy the necessary positivity conditions for all test functions.\footnote{The arguments we made in the main text in favour of taking the `physical algebra' to be ${\cal A}_{0, \mathbb{M}}$ privileged the role of the usual Minkowski ground state (i.e.\ the Poincar\'e invariant vacuum). One might nevertheless still want to explore what could be said about (globally) Hadamard states on the `full' Weyl algebra ${\cal A}_\mathbb{M}$. Within the (technically inequivalent) approach to quantization based on the `full' Borchers-Uhlmann algebra, Schubert  \cite{schubert2013charakterisierung} has recently shown that there are no time-translation invariant Hadamard states; thus it seems reasonable to expect that a similar result will hold within our Weyl-algebra framework. And it seems likely that no boost-invariant Hadamard state exists on ${\cal A}$ either. If so, this would be another reason to take the view that the `physical algebra' is ${\cal A}_0$.}

\subsection{The non-existence theorem}\label{theorem}

Having carefully set up the classical theory for massless fields on our one-mirror spacetime $(M, \eta)$, and having clarified our perspective on both the appropriate strategy to deal with spacetimes with boundaries (in Section \ref{timelikeboundaries}), and the status of the infra-red pathology for massless fields on full (1+1)-dimensional Minkowski spacetime, we now turn to 
the theory obtained by quantizing the classical system analyzed in Section \ref{1+1classicaltheory}.   For this theory, we are now in a position to rigorously define an appropriate class of quantum states for which we are able to prove a non-existence theorem (Theorem \ref{mainthm}) which, arguably (see however Footnote \ref{however}) is analogous to the non-existence result which we conjecture for Kruskal. Namely,  the class of  `strongly boost-invariant globally-Hadamard' states of Definition \ref{stronglyhadamard} below.  Indeed, we will show that once our definitions are in place, the strategy outlined in Section \ref{basicidea} becomes a rigorous proof of this theorem once Equation (\ref{horizontwopoint}) is established. 

In the previous section we have argued that the `physical algebra' for massless fields on full (1+1)-dimensional Minkowski space is the Weyl subalgebra ${\cal A}_{0, \mathbb{M}}$ of ${\cal A}_\mathbb{M}$ generated by Minkowski-space solutions in $S_{0, \mathbb{M}}$. Similarly, here we regard the `physical algebra' for massless fields on $(M, \eta)$, satisfying Dirichlet boundary conditions on $\partial M$, to be \emph{not} ${\cal A} := \mathscr{W}(S, \sigma)$, but rather its subalgebra ${\cal A}_0 := \mathscr{W}(S_0, \sigma)$ generated by solutions in $S_0$ (cf.\ Section \ref{1+1classicaltheory} for definitions of the symplectic vector spaces $(S, \sigma)$ and $(S_0, \sigma)$). 

\begin{defn}\label{stronglyhadamard}
A \emph{strongly boost-invariant globally-Hadamard} state on ${\cal A}_0$ is a boost-invariant state on ${\cal A}_0$ whose spacetime smeared two-point function $\Lambda$ exists and admits an extension $\Lambda^\mathrm{ext}$ to a bidistribution on $\Int M$ which is (i) globally Hadamard in the sense of Definitions \ref{hadamardboundaries} and \ref{hadamard1+1}, (ii) boost-invariant and (iii) a weak bisolution of the wave equation.\footnote{It is not assumed that this extension still satisfies the canonical commutation relations for all test functions, i.e.\ that $\Lambda^\mathrm{ext}(F,G) - \Lambda^\mathrm{ext}(G,F) = iE(F,G)$ for all $F,G \in C_0^\infty(\Int M)$ (of course these are satisfied for pairs of test functions belonging to the subspace $C_{00}^\infty(\Int M)$).}
\end{defn}

We remark that one could contemplate replacing the word `global' in this definition by the word `local' and thereby define a notion of `strongly boost-invariant \emph{locally}-Hadamard'.   However, in view of the fact that no assumption of positivity is made for the extension of the spacetime smeared two-point function, the local-to-global theorem of Radzikowski \cite{radzikowski1996localtoglobal, radzikowski1992hadamard} will presumably not be available to conclude that the two notions are equivalent and it is not clear whether we would be able to prove that there is no state satisfying the local version of the definition.

We point out that, with $\Int M$ replaced by $\R^2$ and ${\cal A}_0$ replaced by ${\cal A}_{0, \mathbb{M}}$ in the above definition, there obviously \emph{is} a strongly boost-invariant globally-Hadamard state on ${\cal A}_{0, \mathbb{M}}$ -- namely $\omega_\mathbb{M}$ as we in fact pointed out at the end of the previous section.  And most importantly, with the obvious replacements, in the case with two mirrors (see the Introduction) there \emph{is} a strongly boost-invariant globally-Hadamard state, namely the `Hartle-Hawking-Israel-like state' constructed in \cite{kay2015instability} with two-point function given by Equation (5) in that paper -- as one may readily verify by inspection of that formula.

\begin{figure}
   \centering
    \includegraphics[scale=0.5]{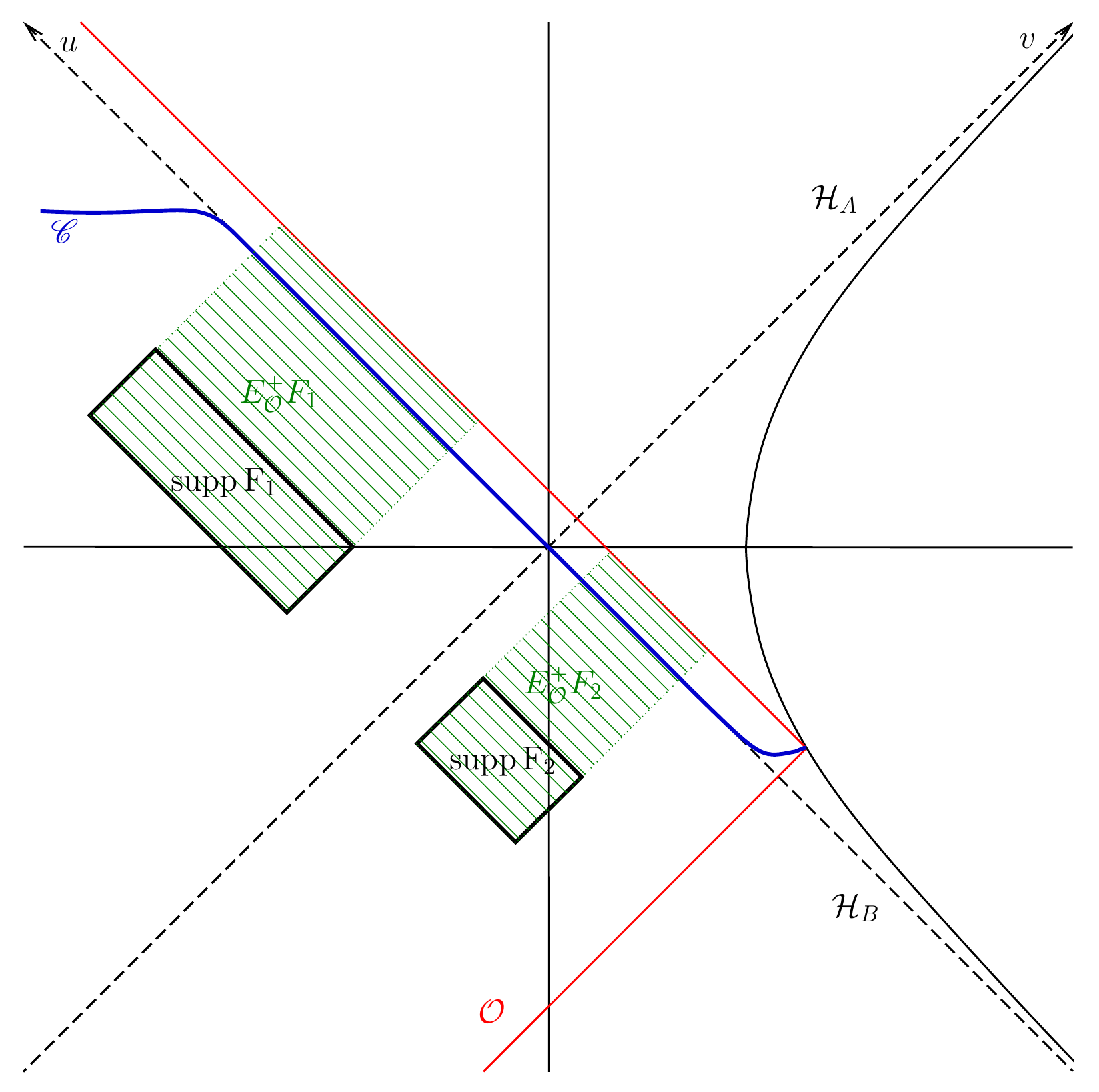}
   \caption{Illustration of the set ${\cal O}$, Cauchy surface $\mathscr{C}$ and test functions $F_1, F_2$ constructed in Lemma \ref{usefullemma} and used in the proof of Theorem \ref{mainthm}. \label{prooffigure}}
 \end{figure}

In contrast, however\ldots

\begin{thm}\label{mainthm}
There is no strongly boost-invariant globally-Hadamard state on ${\cal A}_0$.\footnote{\label{however}This theorem of course implies that there are no boost-invariant states on the `full' Weyl algebra ${\cal A}$ with globally Hadamard spacetime smeared two-point function (in the sense of Definitions \ref{hadamardboundaries} and \ref{hadamard1+1}), since the restriction to ${\cal A}_0$ of any such state would obviously be a strongly boost-invariant globally-Hadamard state on ${\cal A}_0$. However, it does \emph{not} imply that there is no state on ${\cal A}_0$ which is boost-invariant and whose spacetime smeared two-point (exists and) admits an extension to a globally Hadamard bidistribution on $\Int M$, i.e.\ one satisfying (i) but not (ii) and/or (iii) in Definition \ref{stronglyhadamard}.}
\end{thm}

To prove this, we first record and prove a preliminary result in the form of a lemma:

\begin{lem}\label{usefullemma}
For any two solutions $\phi_1, \phi_2$ in $S_B$ one can find a causally convex and globally hyperbolic open subregion $\mathcal{O}$ of $\Int M$, a pair of test functions $F_1, F_2 \in C_{00}^\infty(\mathcal{O})$ and a (partially null) Cauchy surface $\mathscr{C}$ for ${\cal O}$ containing a portion of ${\cal H}_B$, such that
\begin{itemize}
\item the Cauchy data for $\phi_1$ and $\phi_2$ on $\mathscr{C}$ vanish outside $\mathscr{C} \cap {\cal H}_B$; 
\item $F_1$ and $F_2$ have support in $I^-_{\cal O}({\cal H}_B \cap {\cal O}) \cap I^-_{\cal O}(\mathscr{C})$, $E F_1 = \phi_1$ and $E F_2 = \phi_2$.  (Here $I^\pm_{\cal O}(S)$ denotes the chronological future/past of a subset $S$ in ${\cal O}$ \cite{hawking1973large});
\item $E_\mathbb{M} F_1$ is the full Minkowski space solution which is purely right-moving and with restriction to ${\cal H}_B$ equal to $\phi_1\restriction_{{\cal H}_B}$, i.e.\ $E_\mathbb{M} F_1 = T_B \phi_1$ where $T_B$ is the linear symplectomorphism of Proposition \ref{symplectomorphism} (and a similar statement with $F_1 \leftrightarrow F_2$, and $\phi_1 \leftrightarrow \phi_2$).

\end{itemize}
Moreover, $\mathcal{O}$ can be taken to be geodesically convex, and therefore a causal normal neighbourhood of any of its Cauchy surfaces.

\smallskip

\noindent
(All the above of course holds equally with  ${\cal H}_B \leftrightarrow {\cal H}_A$ and $T_B \leftrightarrow T_A$.) 

\end{lem}

\begin{proof}
Since $\phi_i \in S_B$ ($i=1,2$), there exists a unique function $f_i \in C_0^\infty(\R)$ such that $\tilde{\phi_i}(u,v) = f_i(u) - \vartheta(v)f_i(-1/v) \ \forall \ u,v$. Pick $u_\mathrm{m}<0$ with $\supp f_1 \cup \supp f_2 \subset (u_\mathrm{m}, +\infty)$. Then $\mathcal{O} := \left\{ (t, x) \relmiddle| u(t,x) > u_\mathrm{m} , \ v(t,x) < - 1/u_\mathrm{m} \right\}$ is clearly a causally and geodesically convex globally hyperbolic open subregion of $\Int M$, and for $|u_\mathrm{m}|$ sufficiently large it is clear that a Cauchy surface $\mathscr{C}$ for ${\cal O}$ can be found satisfying the requirements in the statement of the Lemma, see Figure \ref{prooffigure}. 

In order to prove the statements about $F_1, F_2$ one proceeds just as in the proof of Theorem \ref{greensoperators} (cf.\ in particular Equations (\ref{causalpropagatorimages})--(\ref{rightleftmovingparts}) and the discussion following these), namely picking any $\chi \in C_0^\infty(\R)$ such that $\supp \chi \subset (-\infty, 0)$ and $\int_\mathbb{\R} \chi(x) \, \d x = 1$, and then defining $\tilde{F_i}(u,v) = -4 f'(u)\chi(v)$.\footnote{Note that, defining $\psi(s) = \int_{- \infty}^{s} \chi(s') \, \d s'$ and $\xi_i(t,x) = - \psi(v(t,x))\phi_i(t,x)$, this amounts to setting $F_i = \square \xi_i$.}
\end{proof}

\begin{proof}[Proof of Theorem \ref{mainthm}] As already outlined in the Introduction and in Section \ref{basicidea}, one need only prove that there are no \emph{quasi-free} strongly boost-invariant globally-Hadamard states. Thus, suppose such a quasi-free state $\omega$ exists with spacetime smeared two-point function $\Lambda : C_{00}^\infty(\Int M) \times C_{00}^\infty(\Int M) \to \mathbb{C}$, and let $\Lambda^\mathrm{ext}$ be an extension of $\Lambda$ satisfying (i), (ii) and (iii) in Definition \ref{stronglyhadamard}. Let $\phi_1, \phi_2 \in S_B$ and pick a causally and geodesically convex, open, globally hyperbolic subset ${\cal O}$ of $\Int M$, a Cauchy surface $\mathscr{C}$ for ${\cal O}$ and a pair of test functions $F_1, F_2 \in C_{00}^\infty({\cal O})$, as in the statement and proof of Lemma \ref{usefullemma}. Since $\Lambda^\mathrm{ext}$ is a globally Hadamard bidistribution on ${\cal O}$, results on the \emph{propagation of the global Hadamard form} contained in \cite{fulling1978singularity, kay1991theorems} guarantee that we are free to choose ${\cal N}={\cal O}$, $\chi \equiv 1$ and $T(t,x)=t$ as the causal normal neighbourhood, `smoothing function' and `global time function' in Definition \ref{hadamard1+1}. In terms of these, the global Hadamard condition for $\Lambda^\mathrm{ext}$ simply reduces to the existence of a function $H_{\cal O} \in C^\infty({\cal O} \times {\cal O})$ such that
\begin{equation}
\Lambda^\mathrm{ext}(F,G) - \Lambda_\mathbb{M}(F,G) = \int\limits_{{\cal O} \times {\cal O}} H_{\cal O}(x, y) F(x) G(x) \, \d \mu_\eta(x)  \, \d \mu_\eta(y)
\end{equation}
for all $F, G \in C_0^\infty({\cal O})$, where $\Lambda_\mathbb{M}$ is as defined in Equation (\ref{minktwopointfunction}). We remark that, since $\Lambda^\mathrm{ext}$ and $\Lambda_\mathbb{M}$ are both weak bisolutions of the wave equation, then $H_{\cal O}$ is a (smooth) bisolution of the wave equation. Also, since both $\Lambda^\mathrm{ext}\restriction_{C_0^\infty({\cal O}) \times C_0^\infty({\cal O})}$ (by assumption) and $\Lambda_\mathbb{M}\restriction_{C_0^\infty({\cal O}) \times C_0^\infty({\cal O})}$ are invariant under the (local) one-parameter group of Lorentz boosts applied to the two copies of $C_0^\infty({\cal O})$ simultaneously, it follows that $H_\mathcal{O}$ is annihilated by the formal adjoint $X^*$ of the infinitesimal generator $X = X_1 \oplus X_2 = (x_1 \partial/\partial t_1 + t_1 \partial/\partial x_1) \oplus (x_2 \partial/\partial t_2 + t_2 \partial/\partial x_2)$ (where, for $i = 1, 2$, $t_i$ and $x_i$ are inertial coordinates on the $i$-th copy of $\mathcal{O}$). Since $X^* = -X$, it follows that $H_\mathcal{O}$ is constant on the integral curves of $X$ on $\mathcal{O} \times \mathcal{O}$. Together with global smoothness (and in particular smoothness at the point $(0,0;0,0)$), this clearly implies that $H_\mathcal{O}$ is constant on the portion of ${\cal H}_B \times {\cal H}_B$ contained within $\mathcal{O} \times \mathcal{O}$.

Now recall that the test functions $F_1$ and $F_2$ were chosen to both have support in $I^-_{\cal O}({\cal H}_B \cap {\cal O}) \cap I^-_{\cal O} (\mathscr{C})$ (see again Figure \ref{prooffigure}). Let $\alpha \in C^\infty({\cal O})$ and $F$ be any test function supported in $I^-_{\cal O} (\mathscr{C})$. Then, proceeding similarly to Equations (B.12)--(B.13) in Appendix B of \cite{kay1991theorems}, and noting that $\square = \nabla^a \nabla_a$,
\allowdisplaybreaks[1]
\begin{align}
\int\limits_{\cal O} \alpha F \, \d \mu_\eta &= \int\limits_{I^-_{\cal O} (\mathscr{C})} \alpha F \, \d \mu_\eta \nonumber \\
&= \int\limits_{I^-_{\cal O} (\mathscr{C})} \alpha \square E^+_{\cal O} F \, \d \mu_\eta \nonumber \\
&= \int\limits_{I^-_{\cal O} (\mathscr{C})} [\square \alpha] E^+_{\cal O} F \, \d \mu_\eta + \int\limits_{I^-_{\cal O} (\mathscr{C})}  \nabla^a[\alpha \overleftrightarrow{\nabla}_a E^+_{\cal O} F] \, \d \mu_\eta \nonumber \\ 
&= \int\limits_{I^-_{\cal O} (\mathscr{C})} [\square \alpha] E^+_{\cal O} F \, \d \mu_\eta + \int\limits_{\mathscr{C}}  n_a [\alpha \overleftrightarrow{\nabla}^a E^+_{\cal O} F] \, \d \mu_{\mathscr{C}} \nonumber \\ 
&=
\int\limits_{I^-_{\cal O} (\mathscr{C})} [\square \alpha] E^+ F \, \d \mu_\eta - \int\limits_{\mathscr{C}}  n_a [\alpha \overleftrightarrow{\nabla}^a E F] \, \d \mu_{\mathscr{C}}\label{greenidentity},
\end{align}
where $E^\pm_{\cal O}$ denotes the retarded/advanced Green operator for $\square$ on ${\cal O}$, in the fourth step Gauss' law has been applied, and in the final step we used the fact that $E^-F$ vanishes on a neighbourhood of $\mathscr{C}$, together with Equation (\ref{Epmrestriction}).

Recalling the fact that $H_{\cal O} \in C^\infty({\cal O} \times {\cal O})$ is a bisolution of the wave equation, and applying Equation (\ref{greenidentity}) twice, with first $\alpha$ interpreted as $\int_{\cal O} H_{\cal O}(\cdot, x_2) F_2(x_2) \, \d \mu_\eta(x_2)$ and $F$ interpreted as $F_1$, and then with $\alpha$ interpreted as $H_{\cal O}(x_1, \cdot)$ for arbitrary fixed $x_1 \in {\cal O}$ and $F$ interpreted as $F_2$, yields
\begin{align*}
\int\limits_{\mathcal{O}\times\mathcal{O}} H_\mathcal{O}(&x_1, x_2) F_1(x_1) F_2(x_2) \, \d \mu_\eta(x_1) \, \d \mu_\eta(x_2) \\
&= \int\limits_{\mathscr{C} \times \mathscr{C}} H_\mathcal{O}(x_1,x_2) \overleftrightarrow{\nabla}^{1a} \overleftrightarrow{\nabla}^{2b} \phi_1(x_1) \phi_2(x_2) \, n_a(x_1) n_b(x_2) \, \d \mu_{\mathscr{C}}(x_1)  \d \mu_{\mathscr{C}}(x_2) \\
&= 4 \int \! \! \! \! \! \! \! \! \! \! \! \! \! \! \! \! \! \! \! \! \! \! \! \! \! \! \mathop{\vphantom{\int}}_{({\cal H}_B \times {\cal H}_B) \cap (\mathcal{O} \times \mathcal{O})} \! \! \! \! \! \! \! \! \! \! [\nabla^{1a} \nabla^{2b} H_\mathcal{O}(x_1,x_2)] \phi_1(x_1) \phi_2(x_2) \, n_a(x_1) n_b(x_2) \, \d \mu_{{\cal H}_B}(x_1)  \d \mu_{{\cal H}_B}(x_2) \\
&= 0. 
\end{align*}
In the second step, we have used the fact that the Cauchy data for $\phi_1$ and $\phi_2$ are supported in $\mathscr{C} \cap {\cal H}_B$ and performed two integrations by parts. The final equality is a consequence of the constancy of $H_\mathcal{O}$ on $({\cal H}_B \times {\cal H}_B) \cap ({\cal O} \times {\cal O})$. This proves that $\Lambda(F_1, F_2) = \Lambda_\mathbb{M}(F_1, F_2)$. In terms of the symplectically smeared two-point function $\lambda_2$ of our state $\omega$, this means that
$$ \lambda_2(\phi_1, \phi_2) = \lambda_{\mathbb{M}}(E_\mathbb{M}F_1, E_\mathbb{M}F_2), $$
where we recall that $\lambda_{\mathbb{M}}$ denotes the symplectically smeared two-point function of the (1+1)-dimensional Minkowski vacuum state $\omega_\mathbb{M}$ on ${\cal A}_{0, \mathbb{M}}$ discussed in Section \ref{infra-red}. But since $F_1$ and $F_2$ were chosen so that $E_\mathbb{M}F_1 = T_B \phi_1$ and $E_\mathbb{M}F_2 = T_B \phi_2$, and since $\phi_1, \phi_2 \in S_B$ are arbitrary, we conclude that in fact
\begin{equation}\label{lambda2lambdam}
\lambda_2(\phi_1, \phi_2) = \lambda_{\mathbb{M}}(T_B \phi_1, T_B \phi_2)
\end{equation}
for \emph{all} $\phi_1, \phi_2 \in S_B$. Next, let $(K, \mathscr{H})$ be the one-particle structure associated to $\omega$, and let $(K_\mathbb{M}, \mathscr{H}_\mathbb{M})$ be the one-particle structure associated to $\omega_\mathbb{M}$ (see Proposition \ref{correspondingoneparticlestructures} in Appendix \ref{appA}), then Equation (\ref{lambda2lambdam}) implies that
\begin{equation}\label{equivalence}
\innerprod{K \phi_1}{K \phi_2}_\mathscr{H} = \innerprod{K_\mathbb{M} T_B \phi_1}{K_\mathbb{M} T_B \phi_2}_{\mathscr{H}_\mathbb{M}}
\end{equation}
(and similarly for $\phi_1, \phi_2 \in S_A$ and $T_A$). Now it is known (cf.\ pages 89--90 in \cite{kay1991theorems}) that

\medskip
\noindent
(A$_\mathbb{M}$) $K_\mathbb{M} S_{\mathrm{r\text{-}mov}}$ and $K_\mathbb{M} S_{\mathrm{l\text{-}mov}}$ are dense in complex-linear subspaces $\mathscr{H}_{\mathrm{r\text{-}mov}}$ and $\mathscr{H}_{\mathrm{l\text{-}mov}}$ of $\mathscr{H}$ (respectively);

\medskip
\noindent (B$_\mathbb{M}$) $K_\mathbb{M} S_{\mathrm{r\text{-}mov}}^R + i K_\mathbb{M} S_{\mathrm{r\text{-}mov}}^R$ is dense in $\mathscr{H}_{\mathrm{r\text{-}mov}}$ and $K_\mathbb{M} S_{\mathrm{l\text{-}mov}}^R + i K_\mathbb{M} S_{\mathrm{l\text{-}mov}}^R$ is dense in $\mathscr{H}_{\mathrm{l\text{-}mov}}$.

\medskip

\noindent But Equation (\ref{equivalence}) immediately implies that the obvious corresponding properties, i.e.\ (A) and (B) of Section \ref{basicidea}, are inherited by $(K, \mathscr{H})$. The proof is then completed exactly as in Section \ref{basicidea}.
\end{proof}

We remark that the connection between the above proof and the heuristic discussion in Section \ref{basicidea} is made clearer if we note that, for any pair $\phi_1, \phi_2 \in S_{\mathrm{r\text{-}mov}}$,
\begin{equation*}
\innerprod{K_\mathbb{M}\phi_1}{K_\mathbb{M}\phi_2}_{\mathscr{H}_\mathbb{M}} = - {1\over \pi} \lim_{\varepsilon\rightarrow 0^+} \int 
{f_1(u_1)f_2(u_2)\over (u_1-u_2-i\varepsilon)^2}\, \d u_1 \d u_2,
\end{equation*}
where $\phi_1(t,x) = f_1(t-x)$ and $\phi_2(t,x) = f_2(t-x)$. Equivalently, $f_1$, $f_2$ can be thought of geometrically as the restrictions of $\phi_1$ and $\phi_2$ (respectively) to ${\cal H}_B$.

\section{Further discussion of the physical relevance of our result}

Our result -- that there is no strongly boost-invariant globally-Hadamard state for the massless wave equation to the left of an eternally uniformly accelerating mirror (with vanishing boundary conditions on the mirror) in 1+1 Minkowski space -- lends support to our conjecture that there is no isometry-invariant Hadamard state for the Klein-Gordon equation defined on the region of Kruskal to the left of a surface of constant Schwarzschild radius in the right wedge (with vanishing boundary conditions on that box).  This suggests that there may be fundamental difficulties in having a semi-classical description of a black hole confined to a spherical static box -- a scenario which is of basic importance in discussions of black hole thermodynamics.  As we discussed in more detail in the Introduction, in an earlier paper \cite{kay2015instability} one of us pointed out a number of senses in which the right wedge horizons become (both classically and quantum mechanically) unstable for the same 1+1 model system with an accelerating mirror and argued for a similar problems for the same Klein-Gordon Kruskal system confined to a box.   The tentative conclusion there was that any semi-classical description in the right wedge must break down at the right-wedge horizons  -- and it was suggested that it makes no sense to consider the spacetime as continuing to have any existence beyond these horizons.

The results of the present paper would seem to lend further support to that conclusion.

One possible way around such a conclusion might be if there were one or more \emph{non-stationary} Hadamard states on the region of Kruskal to the left of the box which are nevertheless stationary when restricted to the region of the right wedge to the left of the box.   If this could be shown to also be impossible it would strengthen the above conclusion further, whereas if it would turn out to be possible it would perhaps undermine the above conclusion.   But it would still be strange if an equilibrium state of a black hole were to be modelled mathematically by a state whose domain of definition includes the left wedge but which is not stationary when restricted to that left wedge.   It would obviously be of interest to try to settle this question -- or the obvious counterpart question for the wave equation in 1+1 Minkowski space to the left of an eternally uniformly accelerating mirror -- but we will not attempt to do so here.

\titleformat{\section}{\Large\bfseries}{\appendixname~\thesection\ --}{0.5em}{}
\begin{appendices}

\section[Appendix A - Weyl quantization of linear systems, quasi-free states and one-particle structures]{Weyl quantization of linear systems, quasi-free states and one-particle structures}\label{appA}

We give here a brief overview of the standard \emph{Weyl-algebra} approach to the quantization of (real, bosonic) linear systems \cite{segal1963mathematical, bratteli1997operator2}. The starting point is the realization that the phase space of the classical theory is a (real) symplectic vector space $(S, \sigma)$. The first step is to construct the \emph{Weyl algebra} \cite{slawny1972factor} \emph{over} $(S, \sigma)$, denoted here by $\mathscr{W}(S, \sigma)$. This is the C$^*$-algebra generated by a unit element $\mathbbm{1}$ and by \emph{Weyl operators} $W(\Phi)$ (for all $\Phi \in S$) satisfying the relations
\begin{equation*}
W(\Phi_1) W( \Phi_2) = e^{-i \sigma(\Phi_1, \Phi_2)/2} W(\Phi_1 + \Phi_2), \qquad
W(\Phi)^* = W(-\Phi),
\end{equation*}
which are to be regarded as exponentiated versions of the standard canonical commutation relations (and in particular imply that each $W(\Phi)$ is unitary and that $W(\boldsymbol{0})=\mathbbm{1}$).

The Weyl algebra construction is functorial in the sense that for any two linear symplectic spaces $(S_1, \sigma_1)$ and $(S_2, \sigma_2)$ and for any linear symplectic map $T : S_1 \to S_2$, one defines in a natural way a *-homomorphism $\alpha : \mathscr{W}(S_1, \sigma_1) \to \mathscr{W}(S_2, \sigma_2)$ between the corresponding Weyl algebras by setting
\begin{equation}\label{inducedhomomorphism} \alpha(W_1(\Phi)) = W_2(T \Phi) \ \forall \ \Phi \in S_1 \end{equation}
(and extending by linearity and continuity). If a one-parameter subgroup $ (\mathscr{T}_\tau)_{\tau \in \R}$ of linear symplectomorphisms of $(S, \sigma)$ is available, then, from the `linear dynamical system' $(S, \sigma, \mathscr{T}_\tau)$, one obtains, via Weyl algebra quantization, the `$C^*$ dynamical system' $({\cal A}, \alpha_\tau)$ where ${\cal A} = \mathscr{W}(S, \sigma)$ and $(\alpha_\tau)_{\tau \in \R}$ is the one-parameter group of *-automorphisms of ${\cal A}$ induced from $(\mathscr{T}_\tau)_{\tau \in \R}$ in the manner described by Equation (\ref{inducedhomomorphism}).

We recall that a \emph{state} on the Weyl algebra ${\cal A}$ is a positive linear functional $\omega$ such that $\omega(\mathbbm{1}) = 1$. It is called \emph{pure} if it cannot be expressed as a convex combination of any other two states, and \emph{mixed} otherwise. Finally, $\omega$ is said to be \emph{stationary} or \emph{invariant} with respect to a one-parameter group $(\alpha_\tau)_{\tau \in \R}$ of *-automorphisms of ${\cal A}$ if, for all $\tau \in \R$, $\omega \circ \alpha_\tau = \omega$.

Correlation functions can be defined for sufficiently regular states; that is, one may define the one- and two-point functions
\begin{align} \lambda_1(\Phi) &= \frac{\d}{\d t} \omega[W(t \Phi)]\Bigg|_{t=0} \\ \lambda_2(\Phi_1, \Phi_2) &= - \frac{\partial^2}{\partial s \partial t} \omega[W(s \Phi_1 + t \Phi_2)]e^{-i s t \sigma(\Phi_1, \Phi_2)/2} \Bigg|_{s, t = 0}, \label{quasifreetwopointfunction} \end{align}
and similarly define higher $n$-point correlation functions $\lambda_n$, if the state is regular enough for the relevant derivatives to exist. Note that all correlation functions are multilinear in their arguments.

Two-point functions play a special role in quantum field theory. For now, note that if a state is $C^2$ (see e.g.\ \cite{kay1993sufficient} for a definition), so that the one- and two-point functions exist, one may verify that $\lambda_2$ automatically satisfies the following properties for all $ \Phi_1, \Phi_2 \in S$:
\begin{enumerate}[(i)] 
\item $\mathrm{Im} [\lambda_2(\Phi_1, \Phi_2)] = \sigma(\Phi_1, \Phi_2)/2$;
\item $\mathrm{Re} \lambda_2 =: \mu$ is a symmetric, real-bilinear form on $S$ satisfying
\begin{equation}\label{positivity}
\mu(\Phi_1, \Phi_1) \geq 0, \qquad \sigma(\Phi_1, \Phi_2)^2 \leq 4 \mu(\Phi_1, \Phi_1) \mu(\Phi_2, \Phi_2).
\end{equation}
\end{enumerate}
Condition (i) encodes the canonical commutation relations, and Condition (ii) results from positivity of the state. 

The set of $\lambda_2 : S \times S \to \C$ satisfying Conditions (i) and (ii) is in one-to-one correspondence with the set of equivalence classes of \emph{one-particle structures} over $(S, \sigma)$, whose definition appeared already in Section \ref{basicidea}, but which we repeat here for convenience.

\begin{defn}[One-particle structures]\label{oneparticlestructures}
These are pairs $(K, \mathscr{H})$, with $\mathscr{H}$ a complex Hilbert space and $K : S \to \mathscr{H}$ a \emph{real-linear} map, such that for all $\Phi_1, \Phi_2 \in S$,
\begin{enumerate}
\item $KS + iKS$ is dense in $\mathscr{H}$;
\item $\mathrm{Im} \innerprod{K \Phi_1}{K \Phi_2}_\mathscr{H} = \sigma(\Phi_1, \Phi_2)/2$.
\end{enumerate}
Any two such pairs $(K, \mathscr{H})$ and $(K', \mathscr{H}')$ are said to be \emph{equivalent} if there exists an isomorphism $U \colon \mathscr{H} \to \mathscr{H}'$ of Hilbert spaces such that $UK = K'$.
\end{defn}

The correspondence works as follows. On the one hand, any one-particle structure $(K, \mathscr{H})$ over $(S, \sigma)$ clearly yields a $\lambda_2$ satisfying Conditions (i) and (ii), namely $\lambda_2(\Phi_1, \Phi_2) = \innerprod{K \Phi_1}{K \Phi_2}_\mathscr{H}$. Somewhat less trivially, the converse also holds.

\begin{prop}\label{correspondingoneparticlestructures}
Given a $\lambda_2: S \times S \to \C$ satisfying Conditions (i) and (ii), there exists a one-particle structure $(K, \mathscr{H})$ which is associated to $\lambda_2$ in the sense that $\innerprod{K \Phi_1}{K \Phi_2}_{\mathscr{H}} = \lambda_2(\Phi_1, \Phi_2)$ for all $\Phi_1, \Phi_2 \in S$. Furthermore, any two such one-particle structures are equivalent in the sense of Definition \ref{oneparticlestructures}.
\end{prop}

The above theorem is proved in Appendix A of \cite{kay1991theorems}. There, and in the discussion following Proposition 3.1 in Section 3.2, it was also pointed out that one may use this result to prove that, for any $\lambda_2 : S \times S \to \C$ satisfying Conditions (i) and (ii) above, the prescription
\begin{equation}\label{gaussian}
\omega [W(\Phi)] = \exp [- \lambda_2(\Phi,\Phi)/2] \quad \forall \ \Phi \in S \end{equation}
(and extension by linearity and continuity) defines a state on ${\cal A}$. Indeed, one may realize the right-hand side of Equation (\ref{gaussian}) as the expectation value in the Fock space vacuum, of the operator $W^\mathscr{F}(K \Phi) = \exp[\overline{a^\dagger(K\Phi) - (a^\dagger(K \Phi))^*}]$ on the Fock space over $\mathscr{H}$. Since $W(\Phi) \mapsto W^\mathscr{F}(K \Phi)$ defines a *-representation of the Weyl algebra, the result follows. One may then easily verify that $\omega$ has a two-point function and that this equals $\lambda_2$. Indeed, $\omega$ also has the following additional properties: (a) it is \emph{analytic} (see e.g.\ \cite{bratteli1997operator2}, p.\ 38) so that, in particular, it is $C^m$ for all $m$ and all correlation functions exist; (b) the one-point function vanishes; (c) the `truncated' $n$-point functions (see e.g.\ \cite{haag1992local, bratteli1997operator2}) vanish for $n > 2$ (in particular, all odd correlation functions vanish). Throughout the present paper, and just as in \cite{kay1991theorems}, we will refer to states having Properties (a)--(c)  as `quasi-free', but remark that more properly they should be referred to as `quasi-free states with vanishing one-point function'. Since analytic states with the same collections of $n$-point functions are identical, this also proves that any quasi-free state on the Weyl algebra is in the form of Equation (\ref{gaussian}), for some $\lambda_2$ satisfying Conditions (i) and (ii).

So one concludes that quasi-free states over the Weyl algebra ${\cal A}$ are also in one-to-one correspondence with equivalence classes of one-particle structures over $(S, \sigma)$, and thus we can freely speak of the (equivalence class of) one-particle structure(s) `associated with' a given quasi-free state. What's more, a number of important properties which could be possessed by a quasi-free state have a `translation' at the level of the corresponding one-particle structure(s). These `one-particle versions' are often technically convenient to work with, and indeed are what allowed us to conjecture/prove the results in the main body of the paper. We record below two such translations (for proofs, see Appendix A of \cite{kay1991theorems} and \cite{kay1985double}), which are invoked in Section \ref{basicidea}. 

\begin{prop}\label{purity}
A state $\omega$ is pure if and only if its associated one-particle structure $(K_\omega, \mathscr{H}_\omega)$ is such that $K_\omega S$ alone is dense in $\mathscr{H}_\omega$.
\end{prop}

\begin{prop}\label{reehschlieder}
Let $\tilde{{\cal A}}$ denote the Weyl algebra over the symplectic vector space $(\tilde{S}, \tilde{\sigma})$ and $\omega$ be a state on $\tilde{{\cal A}}$ with associated one-particle structure $(K_\omega, \mathscr{H}_\omega)$. Then the $C^*$-subalgebra $\tilde{{\cal A}}_R$ of $\tilde{{\cal A}}$ generated by the subspace $R$ of $\tilde{S}$ has the Reeh-Schlieder property\footnote{Let $\omega$ be a state on a $C^*$-algebra $\mathscr{A}$ with GNS-triple \cite{bratteli1987operator1} $(\rho, H, \Omega)$. Then the $C^*$-subalgebra $\mathscr{B}$ of $\mathscr{A}$ is said to have the \emph{Reeh-Schlieder property} for $(\mathscr{A}, \omega)$ if $\rho(\mathscr{B})\Omega$ is dense in $H$.} for $(\tilde{{\cal A}}, \omega)$ iff $K_\omega R + i K_\omega R$ is dense in $\mathscr{H}_\omega$.
\end{prop}

\section[{Appendix B - Filling a gap in [KW91]}]{Filling a gap in \cite{kay1991theorems}}\label{appB}

As already mentioned in Footnote \ref{gap}, we wish here to point out, and attempt to fill, a gap in some of the arguments of \cite{kay1991theorems}.  We should stress that, while it was consideration of our Proposition \ref{symplecticity} which led us to notice this gap, the discussion in this appendix is logically separate from the rest of the paper -- although it may well be that the methods used here will turn out to be useful in the attempt to rigorously prove our no-go conjecture of Sections \ref{intro} and \ref{basicidea} for Kruskal-in-a-box.  As it would be unfeasible to make this appendix fully self-contained, knowledge of the notions, notational conventions (which, by the way, include a different choice of metric signature to the one made in this paper) and general geometric/analytical assumptions underlying the analysis of \cite{kay1991theorems} will be assumed without further comment in what follows and we assume this appendix will be read in conjunction with a copy of \cite{kay1991theorems}.

The gap to be filled is that it is not proven in \cite{kay1991theorems} that the \emph{a priori} \emph{pre}-symplectic subspace $(S_0 = S_A + S_B, \sigma)$ of the symplectic space $(S,  \sigma)$ is also symplectic itself, i.e.\ that $\sigma$ is not only antisymmetric but also non-degenerate on $S_0$.   This gap needs to be filled, in particular, for the proof of Theorem 4.2 in \cite{kay1991theorems} (which, we recall, establishes certain uniqueness and KMS properties) to be valid.  (See where the proof appeals to Lemma 4.1 of \cite{kay1991theorems}.)   As a matter of fact, in view of the issues raised and dealt with in the Note Added in Proof in \cite{kay1991theorems} (cf.\ the discussion in Section \ref{basicidea}), what's \emph{really} important is that a similar job be done on the modification of Theorem 4.2 in that Note Added in Proof involving (a) the `natural extension', to the Weyl algebra $\hat{\mathscr{A}}$ over a suitable larger symplectic space $(\hat{S}, \hat{\sigma})$, of a quasifree, isometry-invariant Hadamard state on the Weyl algebra $\mathscr{A}$ over $(S, \sigma)$, and (b) certain subspaces $\tilde{S}_A$, $\tilde{S}_B$ of $\hat{S}$ which are also suitably `large' subspaces of $S_A$, $S_B$ respectively.  Namely, it needs to be established that the a priori \emph{pre}-symplectic subspaces $(\tilde{S}_A, \hat{\sigma})$, $(\tilde{S}_B, \hat{\sigma})$ and $(\tilde{S}_0 = \tilde{S}_A + \tilde{S}_B, \hat{\sigma})$ of $(\hat{S}, \hat{\sigma})$ are actually symplectic.  In Section \ref{commstartpt} we will give an easy argument, which holds on the entire class of spacetimes considered in \cite{kay1991theorems}, that $(\tilde{S}_A, \hat{\sigma})$ and $(\tilde{S}_B, \hat{\sigma})$ are symplectic.  We will then give two different lines of argument (the first of which applies to the massless Klein-Gordon equation, the second to more general Klein-Gordon equations with isometry-invariant potentials) each of which establishes that $(\tilde{S}_0, \hat{\sigma})$ is symplectic for certain spacetimes with bifurcate Killing horizons which include the notable cases of the Minkowski and Kruskal spacetime.

As explained in the next paragraph but one, both lines of arguments rely in particular on the existence of isometry-invariant Hadamard states for the Klein-Gordon field and spacetime under consideration.    

As we mentioned in the Introduction to this paper, it is also proven in \cite{kay1991theorems} (in Chapter 6)  that, on the globally hyperbolic patches of Schwarzschild-de Sitter (with non-zero Schwarzschild mass) and of sub-extremal Kerr, there is no isometry-invariant Hadamard state.  The same gap needs filling, and we will fill it here, for the validity of the proofs of these non-existence results too in a sense we now explain:    The non-existence proofs assume that the relevant $(\tilde{S}_0, \hat{\sigma})$ are symplectic spaces.   We will show (cf. the previous paragraph) that, if there exists an isometry-invariant Hadamard state for each of these spacetimes, then $(\tilde{S}_0, \hat{\sigma})$ will indeed be symplectic.  Clearly, this suffices to fill the gap in the non-existence proofs, albeit it doesn't actually establish that the $(\tilde{S}_0, \hat{\sigma})$ for these spacetimes is actually symplectic!  We will leave that open.  When we refer, below, to `filling the gap' in the case of Kerr and Schwarzschild-de Sitter, it needs to be borne in mind that this is the sense we intend.

The common starting point for both lines of argument is that, as we will show in Theorem \ref{commstartpoint} in Section \ref{commstartpt}, if 
\begin{enumerate}[label= (\roman*), topsep=1ex, itemsep=-1ex, partopsep=1ex, parsep=1ex]
\item there exists an isometry-invariant Hadamard state on $\mathscr{A}$, and
\item the entire spacetime coincides with the `domain of $C^{k-3}$-determinacy' (with integer $k\geq~5$) of the bifurcate Killing horizon ${\cal H}_A \cup {\cal H}_B$ (this notion will be introduced in Definition \ref{Cndeterminacy}),
\end{enumerate}
then degenerate elements of $(\tilde{S}_0, \hat{\sigma})$\footnote{That is, solutions whose pre-symplectic product with all solutions is zero.} are necessarily `zero modes', i.e.\ are invariant under the isometries.  Once this is established, it immediately follows that $(\tilde{S}_0, \hat{\sigma})$ is symplectic for all those choices of spacetime (with bifurcate Killing horizon) and of Klein-Gordon operator such that (a) Conditions (i) and (ii) above are satisfied, and (b) there do not exist non-zero isometry-invariant solutions in the resulting $\tilde{S}_0$.\footnote{If one were to adopt the fiction explained in Section \ref{basicidea} that $S_A$, $S_B$ and therefore $S_0 = S_A + S_B$ are subspaces of $S$ then there is a simple (though of course false) argument showing that solutions $\phi$ in $S$ which are symplectically orthogonal to the whole of $S_0$ are isometry invariant -- and therefore, apparently, also that $(S_0, \sigma)$ is a symplectic space if there do not exist non-zero isometry-invariant solutions in $S_0$.  This argument does not need to appeal to the existence of any particular quantum state, and therefore Condition (ii) above is not needed. Namely, as explained on page 91 of \cite{kay1991theorems} in the paragraph preceding Lemma 4.1, and under Condition (i) above, in a first step one easily shows that such a $\phi$ must be constant on each null generator of each horizon (we note that, in that passage of \cite{kay1991theorems}, it is stated erroneously that such a solution must be constant on each horizon, but presumably what was intended is what we wrote above).  Then, in virtue of the fact that the isometries map solutions to solutions, and by the very definition of the domain of determinacy, one can conclude that the solution will be isometry-invariant.}  Notice that, as will also be explained in Section \ref{commstartpt}, our definitions of the spaces $\hat{S}$, $\tilde{S}_A$ and $\tilde{S}_B$ (and therefore also $\tilde{S}_0$) will be slightly different from (and, morally speaking, more general than) the ones originally presented in the Note Added in Proof in \cite{kay1991theorems}.  

Sections \ref{decay} and \ref{elliptreg} will present our two different lines of argument which allow to establish the absence of `zero modes' in the cases of interest listed above.  We remark that (a) our methods actually allow to prove a stronger statement, namely the absence of zero modes amongst solutions of sufficient regularity and not just amongst solutions in $\tilde{S}_0,$\footnote{The term `regularity' is here used informally to indicate conditions on both the differentiability and the asymptotic behaviour of the solution.  We will not attempt to precisely identify `minimal' regularity assumptions which are sufficient for ruling out zero modes.} and that (b) since neither of our lines of argument will require that Conditions (i) and (ii) hold, the `cases of interest' \emph{include} Schwarzschild-de Sitter and Kerr.  However, we have not succeeded in ascertaining whether or not there are zero modes in $\tilde{S}_0$ in the case of de Sitter spacetime. 

\subsection{Preliminaries and the common starting point}\label{commstartpt}

Actually, the definition of the `enlarged' symplectic space $(\hat{S}, \hat{\sigma})$ given in \cite{kay1991theorems} is not entirely satisfactory:  $\hat{S}$ is defined there to be the set of real-valued solutions to the Klein-Gordon equation with $C_0^5$ data on a Cauchy surface, $\mathscr{C}$, which contains the entire bifurcation surface $\Sigma$.  It seems not totally clear whether, in this definition, $\mathscr{C}$ is a fixed Cauchy surface, chosen once and for all, or whether it is allowed to depend on the solution.  Either way there would appear to be a serious difficulty:  If the Cauchy surface is allowed to depend on the solution, then there is no reason why $\hat{S}$ should be a vector space.  If it is assumed to be fixed once and for all, then (a statement to the contrary in \cite{kay1991theorems} notwithstanding) there is no reason why the action of the isometries on $(S, \sigma)$ will extend to an action on $(\hat{S}, \hat{\sigma})$.\footnote{This is because the pullback by the isometries of a solution in the thus defined $\hat{S}$ may fail to have $C^5$ Cauchy data on the chosen Cauchy surface $\mathscr{C}$.}

In order to overcome these difficulties we now propose a slightly different candidate for an extension of $(S, \sigma)$ to a larger symplectic space, which we shall also denote $(\hat{S}, \hat{\sigma})$.  We begin by pointing out that  \cite{kay1991theorems} already suggested that an enlarged symplectic space of solutions $\hat{S}$ could alternatively be defined by using Cauchy data on $\mathscr{C}$ belonging to appropriate Sobolev spaces.  In order to turn this idea into a rigorous recipe we will draw upon constructions and results from a recent paper \cite{baer2015initial} by B\"{a}r and Wafo concerning the Cauchy and characteristic initial value problems for an arbitrary second-order normally hyperbolic operator $P$ acting on distributional sections of a vector bundle over a globally hyperbolic spacetime.  To wit, for any choice of foliation of the spacetime by smooth spacelike Cauchy surfaces, the latter authors define spaces of spatially compact solutions to the homogeneous `wave equation' which have `finite $k$-energy' ($k \in \R$) along the foliation, and then show that these spaces do not actually depend on the choice of foliation.  More precisely, given a choice of (smooth) Cauchy temporal function $t : M \to \R$ for the spacetime $M$, one can first define, for each $k \in \R$, spaces $C^\ell(t(M), H^k_{\mathrm{loc}}(\mathscr{C}_\bullet))$ of $\ell$-times continuously differentiable sections of the bundle $\{ H^k_{\mathrm{loc}}(\mathscr{C}_s) \}_{s \in t(M)}$, where $H^k_{\mathrm{loc}}(\mathscr{C}_s)$ is the space of locally Sobolev sections of the  restriction of the original vector bundle to the Cauchy surface $\mathscr{C}_s = t^{-1} \{ s \}$.  As explained in \cite{baer2015initial}, these spaces can then be straightforwardly embedded as subspaces of distributional sections of the original vector bundle over $M$.  It is therefore legitimate to further restrict attention to elements of $C^\ell(t(M), H^k_{\mathrm{loc}}(\mathscr{C}_\bullet))$ which correspond to distributional sections with spacelike-compact support on $M$; this way, one obtains the spaces denoted by $C_{sc}^\ell(t(M), H^k(\mathscr{C}_\bullet))$ in \cite{baer2015initial}.  The space of \emph{finite $k$-energy sections} (with respect to $t$) is then defined by
\begin{equation*}
		\mathscr{F \! E}^k_{sc}(t) = C_{sc}^0(t(M), H^k(\mathscr{C}_\bullet)) \cap C_{sc}^1(t(M), H^{k-1}(\mathscr{C}_\bullet))
\end{equation*}
(this is Definition 1 in \cite{baer2015initial}, though we have suppressed some of the notation there). The main result (which is Corollary 18 in \cite{baer2015initial}) for the purposes of the present Appendix is the fact that, for any two Cauchy temporal functions $t, t'$,
\begin{equation}\label{indeptempfn}
	\mathscr{F \! E}^k_{sc}(t) \cap \ker P = \mathscr{F \! E}^k_{sc}(t') \cap \ker P
\end{equation}
(where we have omitted an embedding into the space of distributional sections from both sides in the interest of notational clarity).  One can thus unambiguously speak of a space $\mathscr{F \! E}^k_{sc}(\ker P)$ of finite $k$-energy \emph{solutions} of the `homogeneous wave equation' for $P$ -- with the property that $\mathscr{F \! E}^k_{sc}(\ker P) = \mathscr{F \! E}^k_{sc}(t) \cap \ker P$ for all Cauchy temporal functions $t$.  Topologizing $\mathscr{F \! E}^k_{sc}(\ker P)$ in the manner discussed in Section 2.7.6 of \cite{baer2015initial}, one has that the spacelike-compact \emph{smooth} solutions of $Pu=0$ form a dense subset of $\mathscr{F \! E}^k_{sc}(\ker P)$.  Furthermore, in a four-dimensional spacetime, by Corollary 20 in \cite{baer2015initial} and the Sobolev embedding theorem, if $\mathbb{N} \ni k \geq 5$ then $\mathscr{F \! E}^k_{sc}(\ker P) \subset C^{k-3}(M) \subset C^{2}(M)$.

In view of the above (and returning to the specific framework of \cite{kay1991theorems}) we define our alternative notion of the space $\hat{S}$ to be one of the spaces $\hat{S}^k = \mathscr{F \! E}^k_{sc}(\ker P)$, with $\mathbb{N} \ni k \geq 5$ to be determined later.  It is then to be understood that, unless stated otherwise, any statement involving `$\hat{S}^k$' (and the later defined `$\tilde{S}_A^k$', `$\tilde{S}_B^k$' and `$\tilde{S}_0^k$') in the remainder of this Appendix will hold for any choice of $\mathbb{N} \ni k \geq 5$.  The denseness of $S$ in $\hat{S}^k$ can be used to show that the `obvious' antisymmetric bilinear form $\hat{\sigma}$ (which we refrain from denoting instead by the more cumbersome `$\hat{\sigma}^k$') on $\hat{S}^k$ is indeed nondegenerate and thus a symplectic form (presumably, a similar density argument was implicitly assumed in \cite{kay1991theorems}).  With our new notion (i.e.\ $\hat{S}^k$) of $\hat{S}$ there is no difficulty in defining a suitable action of the isometries thanks to Equation (\ref{indeptempfn}) together with the fact that the composition of a Cauchy temporal function with an isometry preserving the time orientation yields another Cauchy temporal function.  Finally a quasi-free Hadamard state on the Weyl algebra $\mathscr{A}$ over $(S, \sigma)$ will possess a natural quasi-free extension to the Weyl algebra $\hat{\mathscr{A}}^k$ over $(\hat{S}^k, \hat{\sigma})$ by the same reasoning as in \cite{kay1991theorems}.  We refer to \cite{lupo2015thesis} for more details and rigorous proofs of the statements made in this paragraph.

Just as in \cite{kay1991theorems}, in the case of a Klein-Gordon equation with isometry-invariant potential, spaces of solutions $\tilde{S}_A^k$ and $\tilde{S}_B^k$ can now be defined in such a way that they are at the same time `large' subspaces of $S_A$ and $S_B$ (respectively) and suitable subspaces of $\hat{S}^k$.  Our $\tilde{S}_A^5$ and $\tilde{S}_B^5$ coincide with the $\tilde{S}_A$ and $\tilde{S}_B$ in \cite{kay1991theorems} (respectively).  The key observation, made on pp.\ 133--134 in \cite{kay1991theorems}, is that any function in $S_A$ whose restriction to the $A$-horizon is of the form $\partial^k(U^k g)/\partial U^k$, for some compactly supported and smooth function $g$ on the $A$-horizon, has $C_0^k$ data on any Cauchy surface $\mathscr{C}$ containing the bifurcation surface.  A similar statement (with $U$ replaced by $V$) holds for functions in $S_B$.  Denoting the linear spaces of such solutions by $\tilde{S}_A^k$ and $\tilde{S}_B^k$, this means that, for $k \geq 5$, $\tilde{S}_A^k, \tilde{S}_B^k \subset \hat{S}^k$ as desired.\footnote{\label{distrsolns} We refer to \cite{lupo2015thesis} for details.  We note that, in the aforementioned passage on pp.\ 133--134 in \cite{kay1991theorems}, functions in $S_A$ and $S_B$ are referred to as `solutions'.  However, while they are always continuous, non-zero functions in $S_A$ (resp.\ $S_B$) -- defined in \cite{kay1991theorems} by `gluing' together one-sided solutions to a characteristic initial value problem with data on ${\cal H}_A$ (resp.\ ${\cal H}_B$) -- may fail to be differentiable across ${\cal H}_A$ (resp.\ ${\cal H}_B$).  Hence, they might fail to be classical solutions and one might wonder whether they are solutions even in a weak sense.  That functions in $S_A$ and in $S_B$ are indeed distributional solutions wasn't explicitly shown in \cite{kay1991theorems}, but follows from an application of Gauss' theorem (we thank Alexander Strohmaier for pointing this out to us).}  We also let $\tilde{S}_0^k = \tilde{S}_A^k + \tilde{S}_B^k$.

In order to show that the restriction of $\hat{\sigma}$ to $\tilde{S}_A^k$ is non-degenerate, we now adapt an argument given, in a slightly different context, on page 135 in \cite{kay1991theorems}.  First, recall that if $\phi_1, \phi_2$ are two solutions in $\tilde{S}_A^k$ whose (smooth, compactly supported) restrictions to ${\cal H}_A$ are $f_1$ and $f_2$ respectively, then (cf.\ Equation (4.4) in \cite{kay1991theorems}) one has
\begin{equation}\label{sympformSA}
\hat{\sigma}(\phi_1, \phi_2) = 2 \int_{{\cal H}_A} f_1 \partial_U f_2 \, \sqrt{^{(2)}g} \, \d U \, \d^2 s
\end{equation}
where $\sqrt{^{(2)}g}$ and $\d^2 s$ denote the induced metric and measure on the bifurcation surface.  Suppose now that $\phi$ is a degenerate element in $(\tilde{S}_A^k, \hat{\sigma})$.  Denoting by $f$ the restriction of $\phi$ to the $A$-horizon, it follows, by integrating the right-hand side of Equation (\ref{sympformSA}) by parts $k+1$ times, that $U^k \tfrac{\partial^{k+1} f}{\partial U^{k+1}} = 0$.  Since $f$ is smooth, actually $\tfrac{\partial^{k+1} f}{\partial U^{k+1}} = 0$ everywhere on the horizon.  So $f$ is a polynomial of degree at most $n$ in the affine parameter $U$, whose coefficients are (compactly supported, smooth) functions of the coordinates on the bifurcation surface.  But no such polynomial can have compact support on the $A$-horizon unless it's zero.  This completes the proof that $(\tilde{S}_A^k, \hat{\sigma})$ is a symplectic space.  $(\tilde{S}_B^k, \hat{\sigma})$ is also a symplectic space by a similar argument.

We now turn to what we already called the `common starting point' for both our strategies: That is, we aim to show that, under Conditions (i)--(ii) above, any degenerate element in $(\tilde{S}_0^k, \hat{\sigma})$ is necessarily isometry-invariant.  Before giving a proof of this fact, we must introduce the notion of \emph{domain of $C^n$-determinacy} (with respect to the Klein-Gordon operator) of a subset $U \subseteq M$, with $n \in \mathbb{N} \cup \{\infty\} \cup \{ \omega \}$, which appears in the formulation of our Condition (ii).
\begin{defn}\label{Cndeterminacy}
        The `domain of $C^n$-determinacy' $\mathscr{D}^{(n)}[U]$ (with respect to the Klein-Gordon operator) of $U \subseteq M$ is the set of points $p \in M$ such that every $C^n$ solution which vanishes on $U$ must vanish at $p$. 
\end{defn}
\begin{rk} Kay and Wald's `domain of determinacy' (cf.\ pages 64--65 in \cite{kay1991theorems}) coincides with what we would call the `domain of $C^\infty$-determinacy.'  It is also clear that the inclusions $\mathscr{D}^{(l)}[U] \subseteq \mathscr{D}^{(m)}[U]$ hold for $l \leq m$. \end{rk}

The following Lemma will be used in the proof of the `common starting point', Theorem \ref{commstartpoint} below.

\begin{lem}\label{timetransimpl}
	Let $\omega$ be a quasi-free Hadamard state on $\mathscr{A}$, with associated one-particle structure $(K, \mathscr{H})$.  Let $\hat{K} : \hat{S}^k \to \mathscr{H}$ be the `natural' extension of $K : S \to \mathscr{H}$ \cite{kay1991theorems, lupo2015thesis}.  Then the one-parameter unitary group $U(t)$ on the one-particle Hilbert space $\mathscr{H}$ for $\omega$ which implements the `time translations' $\mathscr{T}(t) : S \to S$ also implements the `time translations' $\hat{\mathscr{T}}(t) : \hat{S}^k \to \hat{S}^k$, i.e.
\begin{equation}\label{tildeimplement}
        U(t) \hat{K} = \hat{K} \hat{\mathscr{T}}(t).
\end{equation}
\end{lem}
\begin{proof}
	For any $\hat{\psi} \in \hat{S}^k$, by definition 
\begin{equation}\label{defnKhat}
\hat{K} \hat{\psi} = \lim_{n \to \infty} K \psi_n
\end{equation}
where $(\psi_n)_{n \in \mathbb{N}}$ is a sequence of solutions in $S$ which converges to $\hat{\psi}$ in the topology for $\hat{S}^k = \mathscr{F \! E}^k_{sc}(\ker P)$ given in \cite{baer2015initial} (and any such sequence yields the same limit on the right-hand-side of Equation (\ref{defnKhat})).  Since $U(t)$ is bounded,
\begin{align*}
	U(t) \hat{K} \hat{\psi} &= U(t)\left[ \lim_{n \to \infty} K \psi_n \right] \\
	&= \lim_{n \to \infty} U(t)[K \psi_n] \\
	&= \lim_{n \to \infty} K [\mathscr{T}(t) \psi_n].
\end{align*}
The claim then follows since it is clear that $(\mathscr{T}(t) \psi_n)_{n \in \mathbb{N}}$ is a sequence in $S$ which tends to $\hat{\mathscr{T}}(t) \hat{\psi}$ in the topology of $\hat{S}^k$.
\end{proof}

We conclude this section with the statement and proof of the `common starting point'.\footnote{Note that it was perhaps suggested in \cite{kay1991theorems} that an even stronger result than Theorem \ref{commstartpoint} should hold, namely that (under the same hypotheses) any solution in $\hat{S}$ (rather than just $\tilde{S}_0$) which is symplectically orthogonal to $\tilde{S}_0$ is isometry-invariant.  However, the integration by parts argument used in the proof of Theorem \ref{commstartpoint} does not straightforwardly adapt in that case, due to the fact that the restrictions of elements in $\hat{S}^k$ to either horizon are in general only in $C^{k-3}$.}

\begin{thm}\label{commstartpoint} Suppose Conditions (i) and (ii) hold.  Then any solution in $\tilde{S}_0^k$ which is symplectically orthogonal to $\tilde{S}_0^k$ is isometry-invariant.
\end{thm}

\begin{proof}
A proof was given in \cite[p.\ 135]{kay1991theorems} (under the unnecessary extra assumption that $k=5$) that if Condition (i) above holds and Condition (ii) is replaced by
\begin{enumerate}[label= (ii$'$), topsep=1ex, itemsep=-1ex, partopsep=1ex, parsep=1ex]
\item the entire spacetime coincides with the domain of $C^\infty$-determinacy of the bifurcate Killing horizon ${\cal H}_A \cup {\cal H}_B$,
\end{enumerate}
then any solution $\phi$ in $S$ with the property that $\hat{\sigma}(\phi, \phi_0) = 0 \ \forall \ \phi_0 \in \tilde{S}_0$ must be isometry-invariant on the entire spacetime.  We now describe how those arguments can be adapted for our purposes.  Let $\psi_0 \in \tilde{S}_0^k$ be such that $\hat{\sigma}(\psi_0, \phi_0) = 0 \ \forall \ \phi_0 \in \tilde{S}_0^k$.  Then, in particular, $\psi_0$ is symplectically orthogonal to the whole of $\tilde{S}_A^k$ and to the whole of $\tilde{S}_B^k$.  We would like to apply an integration by parts argument similar to the one used above in the proof that $(\tilde{S}_A^k, \hat{\sigma})$ and $(\tilde{S}_B^k, \hat{\sigma})$ are symplectic to conclude that the restrictions of $\phi_0$ to ${\cal H}_A$ and ${\cal H}_B$ are polynomials of degree at most $k$ in $U$ and $V$ respectively, whose coefficients in both cases are functions on the bifurcation surface.  However, the restriction of $\phi_0$ to either horizon, while in $C^k$, may fail to be $C^{k+1}$ at the bifurcation surface.  To overcome this difficulty one can apply our integration by parts argument separately, first to symplectic products of $\psi_0$ with solutions in $\tilde{S}_A^{L, k}$ and then to symplectic products of $\psi_0$ with solutions in $\tilde{S}_A^{R, k}$, where
\begin{equation*}
        \tilde{S}_A^{L/R, k} = \left\{ \phi \in \tilde{S}_A^k \mid \text{$\phi$'s } \text{data on ${\cal H}_A$ is of the form } \frac{\partial^k (U^k g)}{\partial U^k} \text{with } g \in C_0^\infty\left({\cal H}_A^{L/R}\right) \right\}
\end{equation*}
and we also define the spaces $\tilde{S}_B^{L/R, k}$ in a similar fashion.  Since the restrictions of $\psi_0$ to ${\cal H}_A^L$ and to ${\cal H}_A^R$ are smooth, it indeed follows that each of them is a polynomial in $U$ of degree at most $k$ whose coefficients are smooth functions on the bifurcation surface.  (The fact that $\psi_0$ is $C^k$ across the bifurcation surface will imply that the first $k$ of these coefficients agree.)  Analogous results clearly hold with $A$ replaced by $B$ and $U$ replaced by $V$.  Now let $\tilde{\mathscr{T}}(t)$ denote the time translation operator on $\tilde{S}_0^k$.  We define a generalized version of the operator $Q(t)$ in Equation (N.4) in \cite{kay1991theorems}, namely
\begin{equation*}
        {}^kQ(t) = \prod_{l = -k}^{k} [\tilde{\mathscr{T}}(t) - e^{l \kappa t}] : \tilde{S}_0^k \to \tilde{S}_0^k.
\end{equation*}
Just as in \cite{kay1991theorems} one sees that since, for any $j$ with $0 \leq j \leq k$, $U^j$ is annihilated by $[\tilde{\mathscr{T}}(t) - e^{j \kappa t}]$, ${}^kQ(t)\psi_0$ vanishes on ${\cal H}_A$.  Similarly, for any $j$ with $0 \leq j \leq k$, $V^j$ is annihilated by $[\tilde{\mathscr{T}}(t) - e^{-j \kappa t}]$, which implies that ${}^kQ(t)\psi_0$ vanishes on ${\cal H}_B$.  Therefore ${}^kQ(t)\psi_0 = 0$ on ${\cal H}_A \cup {\cal H}_B$.  Now, if $\psi_0$ were smooth -- as is $\phi$ in the corresponding arguments in \cite{kay1991theorems} -- the very definition of the domain of ($C^\infty$-)determinacy of a set would immediately imply that, under Condition (ii$'$) above, ${}^kQ(t)\psi_0 = 0$ throughout the spacetime.  However, while $\psi_0$ is certainly everywhere $C^{k-3}$, it could fail to be everywhere smooth.  Thus one cannot conclude that ${}^kQ(t)\psi_0 = 0$ if Condition (ii$'$) alone holds.  However, under the stronger Condition (ii) -- namely under the assumption that the entire spacetime coincides with the domain of $C^{k-3}$ determinacy of the bifurcate Killing horizon -- the vanishing of ${}^kQ(t)\psi_0$ on ${\cal H}_A \cup {\cal H}_B$ does imply that ${}^kQ(t)\psi_0 = 0$ on the entire spacetime.

At this point, again just as in \cite{kay1991theorems}, we invoke Condition (i), i.e.\ the existence of an isometry-invariant Hadamard state on the Weyl algebra $\mathscr{A}$ over $(S, \sigma)$.  Without loss of generality, we can assume this state to be quasi-free and denote its associated one-particle Hilbert space structure by $(K, \mathscr{H})$.  Let $\hat{K} : \hat{S}^k \to \mathscr{H}$ be the `natural extension' of $K : S \to \mathscr{H}$.  By Lemma \ref{timetransimpl}, an equation analogous to Equation (N.6) in \cite{kay1991theorems} holds.  Namely:
\begin{equation}\label{intermediate}
        {}^kP(t) \hat{K} \psi_0 = 0
\end{equation}
where
\begin{equation*}
        {}^kP(t) = \prod_{l=-k}^k [U(t) - e^{l \kappa t}].
\end{equation*}
The desired result that $\psi_0$ is isometry-invariant then follows by straightforwardly adapting the arguments given in the first paragraph on page 136 in \cite{kay1991theorems} (in particular, using in the final step the fact that $\hat{K} : \hat{S}^k \to \mathscr{H}$ is injective, which in turn follows from the property $2 \, \mathrm{Im} \langle \hat{K} \hat{\psi} | \hat{K} \hat{\phi} \rangle = \hat{\sigma}(\hat{\psi}, \hat{\phi}) \ \forall \ \hat{\psi}, \hat{\phi} \in \hat{S}^k$).
\end{proof}

\begin{cor}
 $(\tilde{S}_0^k, \hat{\sigma})$ is a symplectic space if Conditions (i)-(ii) are satisfied and there are no non-zero isometry-invariant solutions in $\tilde{S}_0^k$. \qedsymbol
\end{cor}

We end this section by discussing for which cases of physical interest our Conditions (i) and (ii) are known to hold.  First of all, it is not hard to see that, for any Klein-Gordon equation with isometry-invariant potential, there is no difficulty in adapting the arguments given on pages 64-65 of \cite{kay1991theorems} -- which are based on the characteristic initial value formulations for the sets $J^\pm(\Sigma)$ and on an application of Holmgren's uniqueness theorem -- to our $\mathscr{D}^{(n)}[{\cal H}_A \cup {\cal H}_B]$ for any $n \geq 2$ instead of Kay and Wald's $\mathscr{D}[{\cal H}_A \cup {\cal H}_B]$.  It follows that, for any $k \geq 5$, Condition (ii) holds, for example, on Minkowski spacetime, on the Kruskal spacetime, on de Sitter spacetime, and on the globally hyperbolic patches of Kerr and Schwarzschild-de Sitter considered in \cite{kay1991theorems}.  As for Condition (i), it is known that isometry-invariant Hadamard states exist for the massive or massless Klein-Gordon field on both Minkowski spacetime and \cite{sanders2013thermal} Kruskal spacetime, and for the massive or massless conformally coupled Klein-Gordon field on de Sitter spacetime \cite{chernikov1968quantum, bunch1978quantum}.\footnote{\label{deSittermassless}The case of the massless \emph{minimally coupled} Klein-Gordon field on de Sitter seems more subtle.  While it was proved in \cite{allen1985vacuum} that no fully de Sitter invariant state (Hadamard or not) exists,  Hadamard states do exist \cite{allen1985vacuum, allen1987massless} which are invariant under the subgroups $\mathrm{E}(3)$ and $\mathrm{O}(4)$ of the de Sitter group (and it is presumed \cite{allen1987massless} that $\mathrm{O}(1,3)$-invariant Hadamard states also exist).   However, none of these subgroups contain the `de Sitter boost' isometries to which our analysis applies and we conjecture that there is no boost-invariant Hadamard state.  Our grounds for this conjecture are that, were there to exist such a state, then it is plausible that its restriction to the `right-wedge' (which is of course a static spacetime when the time evolution is taken to be the restriction of the de Sitter boost isometries) would be a KMS state.  But it is known \cite{polarski1990minimally} that (for reasons of bad infra-red behaviour) on the right-wedge, no ground state exists for this time evolution.  Also by Lemma 6.2 in \cite{kay1991theorems}, we know quite generally that if a stationary linear Bose dynamical system admits a KMS state then it also admits a ground state, and thus there would be a contradiction.  There are a number of obstacles, however, to making this argument rigorous:  Even under the fiction explained in Section \ref{basicidea} we would only be able to rely on Theorem 4.2 of \cite{kay1991theorems} to prove the KMS property on the subalgebra of the Weyl algebra for the right wedge corresponding to classical solutions in the subspace of solutions $S_0^R=S_A^R+S_B^R$ and, of course, we don't even know if that theorem is applicable since we don't know if our symplectic form on $S$ restricts to a symplectic form on this subspace.   We also mention, in passing, that since the massless minimally coupled Klein-Gordon field on de Sitter has a classical zero mode (namely the constant solution) the strengthened uniqueness theorem, Theorem 5.1 in Chapter 5 of \cite{kay1991theorems}, is also inapplicable for the reasons explained in the introductory remarks in that Chapter.  We are grateful to Atsushi Higuchi for helpful conversation on the topic of this footnote.}  On the other hand, the paper \cite{kay1991theorems} contains proofs that no such states can exist in Kerr or Schwarzschild-de Sitter, although, as we explained in the fourth paragraph of this Appendix, these proofs have a gap that needs filling and that we will fill below by showing that $(\tilde{S}_0, \hat{\sigma})$ is symplectic under the (for these spacetimes, counter-factual) assumption that an isometry invariant Hadamard state exists.

\subsection{`Decay along the horizons' strategy}\label{decay}

Let us now present our first line of argument for showing the non-existence of isometry-invariant solutions in $\tilde{S}_0^k$.  Note that (see the paragraph below) we are presently only able to to apply this strategy to the case of \emph{massless} fields.  The idea is as follows: suppose $(M,g)$ is a globally hyperbolic spacetime with a bifurcate Killing horizon $\mathcal{H}_A \cup \mathcal{H}_B$ and bifurcation surface $\Sigma$, and suppose that there exists a Cauchy surface $\mathscr{C}$ for $M$ which contains $\Sigma$ and such that $\mathscr{C} = \Sigma \cup (\mathscr{C} \cap \mathscr{L}) \cup (\mathscr{C} \cap \mathscr{R})$, where $\mathscr{L}$ and $\mathscr{R}$ are the left and right wedge regions (respectively) defined in Section 2 of \cite{kay1991theorems}.  Then, clearly (recall that the Killing field is assumed to be complete), an isometry-invariant solution $\phi \in \tilde{S}_0^k$ is identically zero on $M$ if and only if, for all $p \in \mathscr{L} \cup \mathscr{R}$, $\phi(\tau_t (p)) \to 0$ as $t \to + \infty$.  Thus, in the presence of appropriate `pointwise decay' results for (sufficiently regular) solutions of the Klein-Gordon equation in question, the result will follow.

In the case of the massless Klein-Gordon equation, recent papers by Dafermos, Rodnianski and Shlapentokh-Rothman \cite{dafermos2009red, dafermos2010decay, dafermos2014decay, dafermos2007wave} contain pointwise decay results which are sufficient for our purposes in the case of Kruskal and of the globally hyperbolic patches of Kerr and Schwarzschild-de Sitter considered in \cite{kay1991theorems}, \emph{provided} that we pick $k \geq 5$ in the definition of $\tilde{S}_0^k$ large enough for the `higher order weighted energies' defined in those papers to be finite.  That this can always be done can be seen immediately by inspection of the relevant formulae in those papers.\footnote{We also notice that, in the somewhat analogous case of our Proposition \ref{symplecticity}, it is the Dirichlet boundary condition which provides the relevant `decay' for our purposes there.}

\subsection{Strategy based on analytic elliptic regularity}\label{elliptreg}

An alternative approach to showing the non-existence of `zero modes' in $\tilde{S}_0^k$ in a number of important cases, which requires less heavy machinery and also applies to the case where suitable potentials (including e.g.\ a mass term) are included, is based on an application of \emph{analytic elliptic regularity} \cite{john1955plane}.\footnote{We would like to thank Robert Wald for suggesting this approach to us and for providing some guidance on how to deal with the case of Kerr, see below.} Therefore, we must assume the spacetime manifold and metric to be analytic in what follows.

First, we look at the case where the following two conditions hold:
\begin{enumerate}[label= (\alph*), topsep=1ex, itemsep=-1ex, partopsep=1ex, parsep=1ex]
\item the restriction of the spacetime $(M, g)$ and of the one-parameter group of isometries to either the left or the right wedge is analytically isometric to a (globally hyperbolic) \emph{standard static} spacetime (see Section 3.2 of \cite{sanders2013thermal} and references therein) of form $(\mathbb{R}\times C, \alpha \mathrm{d}t^2 -  {}^{3}g)$ where $\alpha$ (the lapse function) is a positive function on $C$ and ${}^3g$ is a Riemannian metric on the connected manifold $C$ (with $C$, $\alpha$ and ${}^3g$ analytic);
\item for any compact set $K \subset M$, the open set $M \setminus J(K)$ has non-empty intersection both with the right and with the left wedge.
\end{enumerate}

It is easy to see that the following spacetimes satisfy the above conditions: Minkowski spacetime with Lorentz boosts as isometries; the Kruskal spacetime with standard `Schwarzschild-time translation isometries'; suitable globally hyperbolic patches of the subextremal Reissner-Nordstr\"{o}m spacetime and of the Schwarzschild-de Sitter spacetime (with non-zero black hole mass), again with their respective standard one-parameter groups of isometries. Importantly, Condition (b) above fails in de Sitter spacetime.  Since scalar fields (massive or massless, conformally or minimally coupled) in de Sitter were also not covered by our other strategy described in Section \ref{decay}, the methods presented in this Appendix are not by themselves sufficient to properly fill the gap in \cite{kay1991theorems} for such fields -- and even the comparatively healthy (see Footnote \ref{deSittermassless}) theory of massive or massless conformally coupled fields in de Sitter spacetime appears to require further investigation.

Under Condition (a), on (say) the right wedge, the Klein-Gordon equation with an analytic potential term $V$ will take the form $\left( \alpha^{-1} \tfrac{\partial^2}{\partial t^2} - {\cal D} + V\right) \phi = 0$ where ${\cal D}$ is the Laplace-Beltrami operator for $^3g$.   For an isometry-invariant solution, $\frac{\partial\phi}{\partial t}$ will be identically zero, and therefore so will be $\frac{\partial^2\phi}{\partial t^2}$ and $\phi$ will satisfy the manifestly elliptic equation with analytic coefficients
\begin{equation*}
\left( - \alpha^{-1} \tfrac{\partial^2}{\partial t^2} - {\cal D} + V\right) \phi = 0
\end{equation*}
(and we notice that the operator $-\alpha^{-1} \tfrac{\partial^2}{\partial t^2} +  {\cal D}$ is of course nothing but minus the Laplace-Beltrami operator for the Riemannian metric $\alpha \d t + {}^{3}g$). Therefore, by analytic elliptic regularity, $\phi$ must be an analytic function on the right wedge. But, since $\phi \in \tilde{S}_0 \subset \hat{S}$ has Cauchy data -- on a Cauchy surface $\mathscr{C}$ for the full spacetime which contains the bifurcation 2-sphere $\Sigma$, see \cite{kay1991theorems} -- of compact support, by finite propagation speed results it must vanish on $M \setminus J(K)$ where $K = \supp (\phi\restriction_{\mathscr{C}}) \cup \supp (\nabla_n\phi\restriction_{\mathscr{C}})$ and $n$ denotes the vector field of unit normals to $\mathscr{C}$. Under Condition (b) above, it must then vanish in an open subset of the right wedge. By analyticity and connectedness, it must vanish identically on the entire right wedge. A similar argument shows that it must vanish identically on the left wedge. Finally, $\phi$ must vanish on the entire spacetime by continuity at $\Sigma$. 

An obvious local-to-global version of this argument also shows that the same conclusion holds if we only require, in Condition (a) above, that the spacetime in the left and right wedges be simply static (with respect to the one-parameter group of isometries) rather than `standard' static. However, outside these circumstances the argument won't straightforwardly apply. Nonetheless, under some mild restrictions on the possible potential terms which we shall state, one can also fill the gap for the case of the globally hyperbolic patch of (subextremal, maximally extended) Kerr defined on page 66 of \cite{kay1991theorems} and denoted by $\mathscr{M}$ there, with Killing vector field $\xi_+ = \partial/\partial t + \Omega_+ \partial/\partial \varphi$ in Boyer-Lindquist coordinates $(t,r,\theta,\varphi)$. Here, denoting the black hole's angular momentum by $a$ and its mass by $M$, $\Omega_+ = a/(r_+^2 + a^2)$ is the angular velocity of the black  hole/Killing horizon situated at $r=r_+=M + \sqrt{M^2 - a^2}$ and we recall that there is a cosmological horizon `at' $r=r_-=M - \sqrt{M^2 - a^2}$. In the right wedge where the Boyer-Lindquist coordinates are regular, the Laplace-Beltrami operator associated with the Kerr metric is
\begin{equation*}
\square = \left[ a^2 \sin^2 \theta - \frac{(a^2 + r^2)^2}{\Delta(r)} \right] \frac{\partial^2}{\partial t^2} - \frac{a^2}{\Delta(r)} \frac{\partial^2}{\partial \varphi^2} - \frac{2a[r^2 + a^2 - \Delta(r)]}{\Delta(r)} \frac{\partial^2}{\partial \varphi \partial t} + \frac{\partial}{\partial r} \left[ \Delta(r) \frac{\partial}{\partial r}\right] + \slashed{\Delta}_{\mathbb{S}^2},
\end{equation*}
where $\Delta(r) = (r-r_+)(r-r_-)$ (so that $\Delta(r)>0$ everywhere in the right wedge) and $\slashed{\Delta}_{\mathbb{S}^2} = \tfrac{1}{\sin \theta} \tfrac{\partial}{\partial \theta} \left( \sin \theta \tfrac{\partial}{\partial \theta} \right) + \tfrac{1}{\sin^2 \theta} \tfrac{\partial^2}{\partial \varphi^2}$ is the Laplacian on the two-dimensional unit sphere. Now, let $u$ be a $C^2$ function on $\mathscr{M}$ which is invariant under the isometries generated by $\xi_+$. Then, everywhere in the right wedge,
\begin{equation}\label{isominvkerr}
\frac{\partial u}{\partial \varphi} = -\Omega_+^{-1} \frac{\partial u}{\partial t} \quad \text{and} \quad \frac{\partial^2 u}{\partial \varphi^2} = \Omega_+^{-2} \frac{\partial^2 u}{\partial t^2}.
\end{equation}
Thus, if $u$ is an isometry-invariant solution to $\square u = 0$ on $\mathscr{M}$, belonging to $\tilde{S}_0$, then we can use the equations in (\ref{isominvkerr}) to `trade' $\varphi$-derivatives for $t$-derivatives and obtain
\begin{equation}
        \left\{ F(r, \theta) \frac{\partial^2}{\partial t^2} + \frac{\partial}{\partial r} \left[ \Delta(r) \frac{\partial}{\partial r}\right] + \frac{1}{\sin \theta} \frac{\partial}{\partial \theta} \left[ \sin \theta \frac{\partial}{\partial \theta} \right] \right\} u = 0
\end{equation}
where $F(r, \theta)$ is an analytic function for $(r, \theta) \in (r_+, \infty) \times (0, \pi)$. Clearly, the same equation will be satisfied by the Fourier coefficients
\begin{equation*}
\hat{u}_m(t,r,\theta) := \int_0^{2 \pi} u(t,r,\theta,\varphi) e^{-i m \varphi} \, \d \varphi, \quad m \in \mathbb{Z}.
\end{equation*}
However, a simple calculation shows that, in virtue of the first equation in (\ref{isominvkerr}),
\begin{equation*}
\frac{\partial \hat{u}_m}{\partial t} = i m \Omega_+ \hat{u}_m \quad \text{and, consequently,} \quad \frac{\partial^2 \hat{u}_m}{\partial t^2} = - m^2 \Omega_+^2 \hat{u}_m
\end{equation*}
for all $m \in \mathbb{Z}$. Pick a positive constant $K$ and set $G(r, \theta) = - m^2 \Omega_+^2 [F(r, \theta) - K ]$; then $\hat{u}_m$ solves $P \hat{u}_m = 0$ where
\begin{equation*}
        P = K \frac{\partial^2}{\partial t ^2} + \Delta(r) \frac{\partial^2}{\partial r^2} + \frac{\partial^2}{\partial \theta^2} + \frac{\d \Delta}{\d r}(r) \frac{\partial}{\partial r} + \cot \theta \frac{\partial}{\partial \theta} + G(r,\theta).
\end{equation*}
$P$ is a differential operator with analytic coefficients. An inspection of the highest order terms shows that it is elliptic on $\R \times (r_+, \infty) \times (0,\pi)$. Therefore, by analytic elliptic regularity, $\hat{u}_m$ is analytic. But $\hat{u}_m$ must vanish in an open set because of the support properties of $u \in \tilde{S}_0$.\footnote{To quickly see this, the reader may wish to consider a \textit{projection diagram}, in the sense of \cite{chrusciel2012space-time} (see also Chapter 3 of \cite{chrusciel2015geometry}), for the region of Kerr under consideration and denoted by $\mathscr{M}$ above. The projection diagram in Fig.\ 3 of \cite{chrusciel2012space-time} appears to closely resemble the more commonly seen conformal diagram for the submanifold corresponding to the axis of symmetry ($\theta = 0$ or $\theta = \pi$) of the Kerr solution. However, unlike the latter, the former captures causal properties of the entire spacetime (in a precise way discussed in Section 3 of \cite{chrusciel2012space-time}). In particular, since $u$ above has spacelike compact support on $\mathscr{M}$, it follows that the projection of its support onto the (1+1)-dimensional diagram is spacelike compact with respect to the (1+1)-dimensional Minkowski metric. The claimed result then easily follows upon observing that the projection diagram is obtained by projecting out the Boyer-Lindquist coordinates $\theta$ and $\varphi$.} Therefore $\hat{u}_m=0$ for all $m \in \mathbb{Z}$. By the Fourier inversion formula, this in turn implies that $u=0$ in the right wedge. Similar reasoning shows that $u$ must vanish in the left wedge. Again, by continuity at the bifurcation surface this means that $u$ must vanish on $\mathscr{M}$. For ease of presentation, we only showed the proof explicitly in the case of the massless wave equation. However, it is clear that an analytic potential term can be added with no change in the arguments, provided it is independent of the coordinate $\varphi$ -- as would of course be the case for a constant mass term or for a constant multiple of the Ricci scalar.

\subsection{Conclusions}

To conclude, the two lines of argument presented in this Appendix have enabled us to fill the gap in \cite{kay1991theorems} in many cases of interest (however, see our discussion in the introductory section of this Appendix for the meaning of `filling the gap' in the cases of Schwarzschild-de Sitter and Kerr).  In the case of de Sitter spacetime, it is not obvious to us that there can be no non-zero solutions in $\tilde{S}_0$ which are invariant under the one-parameter group of isometries generating the bifurcate Killing horizon. Clearly, for massless minimally coupled fields, there \emph{are} non-zero solutions in $\hat{S} \supset \tilde{S}_0$ which are invariant: namely, the constant (non-zero) solutions -- therefore, in particular, one would need to show that no non-zero constant solution can lie in $\tilde{S}_0$. But this would still not suffice to fill the gap.

\end{appendices}

\section*{Acknowledgements}

U.L. gratefully acknowledges the financial support of the Department of Mathematics, University of York, by the award of a teaching studentship. We would also like to thank Igor Khavkine for pointing out the reference \cite{gindikin1996mixed}, Thomas-Paul Hack for pointing out the reference \cite{schubert2013charakterisierung}, Robert Wald for collaboration on the issue of the `gap' in \cite{kay1991theorems} dealt with in Appendix \ref{appB}, Mihalis Dafermos for pointing out that Theorem 10 in \cite{dafermos2014scattering} substantiates our claims about Property (b) in the Introduction, Atsushi Higuchi for helpful discussions concerning massless fields in de Sitter, and Alexander Strohmaier for suggesting the argument given in Footnote \ref{distrsolns}.

\bibliographystyle{hep}
\bibliography{mybibliography}

\end{document}